%% file: arxiv.tex
\documentclass{article}

\usepackage[a4paper, total={6in, 8in}]{geometry}
\usepackage{booktabs} 

\usepackage{authblk}
\usepackage{algorithmic}
\usepackage{amsmath}
\usepackage{nicefrac}
\usepackage{amsthm}
\usepackage{dsfont}
\usepackage{amsfonts}
\usepackage{amssymb}
\usepackage{soul,color}
\usepackage{float}
\usepackage{paralist}
\usepackage[caption = false]{subfig}
\usepackage{graphicx}
\usepackage{url}
\usepackage{paralist}
\usepackage{microtype}

\DeclareMathOperator*{\argmax}{arg\,max}

\newcommand{\algoname}{\textsc{GT-UCB}}
\newcommand{\algonamefat}{\textsc{Fat-GT-UCB}}

\newtheorem{problem}{Problem}
\newtheorem{definition}{Definition}
\newtheorem{lemma}{Lemma}
\newtheorem{theorem}{Theorem}

\usepackage[draft,inline,nomargin,marginclue]{fixme} 
\fxsetup{theme=color,mode=multiuser}
\FXRegisterAuthor{oc}{aoc}{Olivier}

\usepackage[ruled]{algorithm2e} 

\SetAlFnt{\small}
\SetAlCapFnt{\small}
\SetAlCapNameFnt{\small}
\SetAlCapHSkip{0pt}
\IncMargin{-\parindent}

\begin{document}
\title{Algorithms for Online Influencer Marketing}

\author[1]{Paul Lagr\'ee}
\author[2]{Olivier Capp\'e}
\author[1]{Bogdan Cautis}
\author[1]{Silviu Maniu}
\affil[1]{LRI, Universit\'e Paris-Sud, Universit\'e Paris-Saclay}
\affil[2]{LIMSI, CNRS, Universit\'e Paris-Saclay}

\date{}
\maketitle

\begin{abstract}
  Influence maximization is the problem of finding influential users, or nodes,
  in a graph so as to maximize the spread of information. It has many
  applications in advertising and marketing on social networks. In this paper,
  we study a highly generic version of influence maximization, one of optimizing influence campaigns
  by sequentially selecting ``spread seeds''  from a set of \emph{influencers}, a
  small subset of the node population, under the hypothesis that, in a given
  campaign, previously activated nodes remain ``persistently'' active throughout
  and thus do not yield further rewards.  This problem is in particular relevant for an important form of online marketing, known as \emph{influencer marketing}, in which the marketers target a sub-population of influential people, instead of the entire base of potential buyers.  Importantly, we make no assumptions on
  the underlying diffusion model and we work in a setting where neither a
  diffusion network nor historical activation data are available.  We call this
  problem  \emph{online influencer marketing with persistence} (in short, OIMP). We first
  discuss motivating scenarios and present our general approach.
  We introduce an estimator on the influencers' \emph{remaining potential} --
  the expected number of nodes that can still be reached from a given 
  influencer -- and justify its strength to rapidly estimate the desired value,
  relying on real data gathered from Twitter.  We then describe a novel
  algorithm, \algoname, relying on upper confidence bounds on the remaining
  potential. We show that our approach leads to high-quality spreads on both
  simulated and real datasets, even though it makes almost no assumptions on the
  diffusion medium.  Importantly, it is orders of magnitude faster than
  state-of-the-art influence maximization methods, making it possible to deal with large-scale
  online scenarios.
\end{abstract}

%
%
%
%



\section{Introduction}
\label{sec:introduction}

\input{introduction}

\section{Setting}
\label{sec:problem}

\input{problem}

\section{Algorithm}
\label{sec:algorithm}

\input{algorithm}

\section{Analysis}
\label{sec:analysis}

\input{analysis}

\section{OIMP with influencer fatigue}
\label{sec:variant}

\input{variant-short}

\section{Experiments}
\label{sec:experiments}

\input{experiments}

\section{Other related work}
\label{sec:related}

\input{related}

\section{Conclusion}
\label{sec:conclusion}

\input{conclusion}

\section*{Acknowledgments}
This work was partially supported by the French research project ALICIA
(grant ANR-13-CORD-0020).

\bibliographystyle{plain}
\bibliography{biblio-small-kdd} 

\appendix
\input{arxivappendix}

\end{document}

%% file: introduction.tex

Advertising based on word-of-mouth diffusion in social media has become very
important in the digital marketing landscape. Nowadays, social value and social
influence are arguably the hottest concepts in the area of Web advertising and
most companies that advertise in the Web space must have a ``social'' strategy.
For example, on widely used platforms such as Facebook or Twitter, promoted
posts are interleaved with normal posts on user feeds. Users interact with these
posts by  actions such as ``likes'' (adoption), ``shares'' or ``reposts''
(network diffusion). This represents an unprecedented utility in advertising, be it
with a commercial intent or not,  as  products, news, ideas, movies, political
manifests, tweets, etc, can propagate easily to a large
audience~\cite{watts2003,watts07}.


Motivated by the need for effective viral marketing strategies, \emph{influence
estimation} and \emph{influence  maximization} have become important
research problems, at the intersection of data mining and social
sciences~\cite{easley10}.  In short, influence maximization is the problem of selecting a set of
nodes from a given diffusion graph, maximizing the expected spread under an
underlying diffusion model. This problem was introduced in $2003$ by the seminal
work of Kempe et al.~\cite{kempe03}, through two stochastic, discrete-time
diffusion models, \textit{Linear Threshold} (LT) and \textit{Independent
Cascade} (IC). These models rely on diffusion graphs whose edges are weighted by
a score of influence.  They show that selecting the set of nodes maximizing the
expected spread is NP-hard for both models, and they propose a greedy algorithm
that takes advantage of the sub-modularity property of the influence spread, but
does not scale to large graphs.  A rich literature followed, focusing on
computationally efficient and scalable algorithms to solve influence
maximization. The recent
benchmarking study of Arora et al.~\cite{arora17} summarizes state-of-the-art
techniques and also debunks many influence maximization myths.  In particular, it shows that,
depending on the underlying diffusion model and the choice of  parameters, each
algorithm's  behavior can vary significantly, from very efficient to
prohibitively slow, and that influence maximization at the scale of real applications remains an elusive target.  

Importantly, all the influence maximization studies discussed in~\cite{arora17} have as starting
point a specific diffusion model (IC or LT), whose  graph topology and
parameters -- basically the edge weights -- are known. In order to \emph{infer}
the diffusion parameters or the underlying graph structure, or
both,~\cite{gomez12,gomez13,gomez11,goyal10,saito08,du13} propose
\emph{offline, model-specific methods}, which rely on observed information
cascades. In short, information cascades are time-ordered sequences of records
indicating when a specific user was activated or adopted a specific item.
%

There are however many situations where it is unreasonable
to assume the existence of relevant historical data in the form of cascades. For
such settings, \emph{online approaches}, which can learn the underlying
diffusion parameters \emph{while running diffusion campaigns}, have been
proposed. Bridging influence maximization and inference, this is done by balancing between
exploration steps (of yet uncertain model aspects) and exploitation ones (of the
best solution so far), by so called \emph{multi-armed bandits} techniques, where
an agent interacts with the network to infer influence
probabilities~\cite{vaswani15,chen16-2,wen16,vaswani17}. The learning agent sequentially
selects seeds from which diffusion processes are initiated in the network; the
obtained feedback is used to update the agent's knowledge of the model.

Nevertheless, all these studies on inferring diffusion networks, whether offline
or online, rely on parametric diffusion models, i.e., assume that the actual
diffusion dynamics are well captured by such a model (e.g., IC). This maintains
significant limitations for practical purposes. First, the more complex the
model, the harder to learn in large networks, especially in campaigns that have a relatively short
timespan, making  model inference and parameter estimation very challenging
within a small horizon (typically tens or hundreds of spreads). Second,  it is
commonly agreed that the aforementioned  diffusion models represent elegant yet
coarse interpretations of a reality that is much more complex and often hard to
observe fully.  For  examples of insights into this complex reality, the
\emph{topical} or \emph{non-topical} nature of an influence campaign, the
\emph{popularity} of the piece of information being diffused, or its  specific
\emph{topic}  were all shown to have a significant impact on hashtag diffusions
in Twitter~\cite{du13,grabowicz16,romero11}.




\paragraph*{Our contribution} Aiming to address such limitations, we propose in
this paper a \emph{large-scale approach for online and adaptive influence
maximization}, in which
the underlying assumptions for the diffusion processes are kept to a minimum
(if, in fact, hardly any). 

We argue that it can represent a versatile tool in
many practical scenarios, including in particular the one of \emph{influencer marketing}, which, according to Wikipedia, can be described as follows: \emph{``a form of marketing in which focus is placed on influential people rather than the target market as a whole, identifying the individuals that have influence over potential buyers, and orienting marketing activities around these influencers''}.  For instance,  influential users may  be contractually bound by a sponsorship, in exchange for the publication of promoted posts on their online Facebook or Twitter accounts. 
This new form of marketing is by now  extensively used in online social platforms, as is discussed in the marketing literature~\cite{gillin07,duncan08, brown13}.

More concretely, we focus on social media diffusion
scenarios  in which influence campaigns consist of multiple \emph{consecutive
trials} (or \emph{rounds}) spreading the same type of information from an
arbitrary domain (be it a product, idea, post, hashtag, etc)\footnote{Repeated
exposure, known also as the ``effective frequency'',  is a crucial concept in
marketing strategies, online or offline.}. The goal of each campaign is to reach
(or \emph{activate}) as many distinct users as possible, the objective function
being the total spread. These potential influencees -- the target market -- represent the nodes of an unknown diffusion medium / online network. 
In our setting -- as, arguably, in many real-world scenarios -- the
campaign selects from a known set of spread seed candidates, so called \emph{influencers}, a small subset of
the potentially large and unknown target market. At each round, the learning agent
picks among the influencers those from which a new diffusion process is
initiated in the network, gathers some feedback on the activations, and adapts
the subsequent steps of the campaign. Only the effects of the diffusion process, namely the activations (e.g., purchases or subscriptions),  are observed, but not the process itself.  The agent may ``re-seed'' certain influencers
(we may want to ask a particular one to initiate spreads several times, e.g.,
if it has a strong \emph{converting impact}). This perspective on influence
campaigns imposes naturally a certain notion of \emph{persistence}, which is
given the following interpretation: users that were already activated in the
ongoing campaign -- e.g., have adopted a product or endorsed a political
movement --  remain activated throughout that campaign, and thus will not be
accounted for more than once in the objective function.

We call this problem \emph{online influencer marketing with persistence} (in
short, OIMP). Our solution for it follows the multi-armed bandit idea initially
employed in Lei et al.~\cite{lei15}, but we adopt instead a
\emph{diffusion-independent perspective}, whose only input are the spread seed
candidates, while the population and underlying diffusion network -- which may
actually be the superposition of several networks -- remain unknown.   In our
bandit approach, the parameters to be estimated are the values of the
influencers --~how good is a specific influencer~--, as opposed to the diffusion
edge probabilities of a known graph as in~\cite{lei15}.  Furthermore, we 
make the model's feedback more realistic by assuming that
after each trial, the agent only gathers the set of activated nodes.
The rationale is that oftentimes, for a given ``viral'' item, we can track in
applications only \emph{when} it was adopted by various users, but not
\emph{why}. A key difference w.r.t. other multi-armed bandit studies for
influence maximization such as~\cite{vaswani15,chen16-2,wen16,vaswani17} is that
these look for a \textit{constant} optimal set of seeds, while the difficulty
with OIMP is that the seemingly best action at a given trial depends on the
activations of the previous trials (and thus the learning agent's past
decisions).

The multi-armed bandit algorithm we propose, called \algoname, relies on a
famous statistical tool  known as the \emph{Good-Turing estimator}, first
developed during WWII to crack the Enigma machine, and later published by Good
in a study on species discovery~\cite{good53}.  Our approach is inspired by the
work of Bubeck et al.~\cite{bubeck13}, which proposed the use of the Good-Turing
estimator in a context where the learning agent needs to sequentially select
experts that only sample one of their potential nodes at each trial. In
contrast, in OIMP, when an influencer is selected, it may have a potentially large
spread and may activate many nodes at once.  Our solution follows the well-known
\textit{optimism in the face of uncertainty} principle from the bandit
literature (see~\cite{bubeck12} for an introduction to multi-armed bandit
problems), by deriving an upper confidence bound on the estimator for the
remaining potential for spreading information of each influencer, and by choosing
in a principled manner between explore and exploit steps.

In Section~\ref{sec:experiments} we evaluate the proposed approach on
 publicly available graph datasets as well a
large snapshot of Twitter activity we collected. The proposed algorithm is
agnostic with respect to the choice of influencers, who in most realistic
applications will be selected based on a combination of graph-based and external considerations. We
describe however in Section~\ref{sec:influencers} several heuristics that may be used to automatically extract good
candidates for the influencer set, when a social network graph is
available. This choice of influencers also makes it possible to compare
the results obtained by our method to those of the baseline methods from
the literature (see above), which require knowledge of the graph and
of the diffusion model and, in some cases, of the influence
probabilities.

\paragraph{Comparison with previous publication} We extend in this article a preliminary study published in Lagr\'ee et al.~\cite{lagree17}, introducing the following new contributions, which allow us to give a complete picture on our model and algorithmic solutions: 
\begin{itemize} 
\item A more detailed discussion on the motivation behind our work, in relation to new forms of online marketing, such as influencer marketing. 
\item An empirical analysis over Twitter data, which comes to support the assumptions behind our choice of estimators for the remaining potential in each influencer.
\item A detailed theoretical analysis and justification for the upper confidence bounds on which our algorithm \algoname\ relies.
\item Theoretical guarantees on the performance of \algoname, formulated in terms of \emph{waiting time}, a notion that is better suited  in our bandit framework than the usual one of \emph{regret}.
\item An adaptation of  \algoname\ (denoted \algonamefat) and the corresponding theoretical analysis, for scenarios in which influencers may experience fatigue, i.e., a diminishing tendency to activate their user base as they are re-seeded throughout a marketing campaign.   
\item A broader experimental analysis involving the two datasets previously used in~\cite{lagree17}, as well as a new set of experimental results in a completely different scenario, involving real influence spreads from Twitter, for both \algoname\ and \algonamefat.  
\end{itemize}

To the best of our knowledge, our approach is the first to show that efficient and effective influence
maximization 
can be  done in a highly uncertain or under-specified social environment, along with formal guarantees on the achieved spread.

%% file: problem.tex

The goal of the \textit{online influencer marketing with persistence} is to
successively select (or \emph{activate}) a number of seed nodes,
~in order to \emph{reach} (or \emph{spread} to) as
many other nodes as possible. In this section, we formally define this problem.

\subsection{Background}

Given a graph $G = (V, E)$, the traditional problem of influence maximization
is to select a set of seed nodes $I \subseteq V$, under a cardinality
constraint $|I|=L$, such that the expected \emph{spread} --~that is, the number of activated nodes~--
of an influence cascade starting from $I$  is
maximized. Formally, denoting by the random variable $S(I)$ the spread
initiated by the seed set $I$, influence maximization aims to solve the
following optimization problem:

\[
  \argmax_{I \subseteq V, |I|=L} \mathbb{E}[|S(I)|].
\]

As mentioned before, a plethora of algorithms have been proposed to solve the
influence maximization problem, under specific diffusion models. These
algorithms can be viewed as \emph{full-information} and \emph{offline}
approaches: they choose all the seeds at once, in one step, and they have the
complete diffusion configuration, i.e., the graph topology and the influence
probabilities.

In the \emph{online} case, during a sequence of $N$ (called hereafter the
\emph{budget}) consecutive trials, $L$ seed nodes are selected at each trial,
and \emph{feedback} on the achieved spread from these seeds is collected.
%


\subsection{Influence maximization via influencers}

The short timespan of campaigns makes parameter estimation very challenging
within small horizons. In other cases, the topology -- or even the existence --
of a graph is too strong an assumption. In contrast to~\cite{lei15}, we do not
try to estimate edge probabilities in some graph, but, instead, we assume the
existence of a known set of spread seed candidates -- in the following referred to as the \emph{influencers} -- who are the only access to
the medium of diffusion.  Formally, we let $[K] := \{1, \ldots, K\}$ be a set of
influencers up for selection; each influencer is connected to an unknown and
potentially large base (the influencer's \emph{support}) of basic nodes, each
with an unknown activation probability.  For illustration, we give in
Figure~\ref{fig:depth1} an example of this setting, with $3$ influencers
connected to $4$, $5$, and $4$ basic nodes, respectively.

Now, the problem boils down to estimating the value of the $K$ influencers,
which is typically much smaller than the number of parameters of the diffusion
model.  The medium over which diffusion operates may be a diffusion
graph but we make no assumption on that, meaning that the
diffusion may also happen in a completely unknown environment. Finally, note that by
choosing $K = |V|$ influencers, the classic influence maximization problem can
be seen as a special instance of our setting.



\begin{figure}[t]
	\centering
	\includegraphics[height=3.3cm]{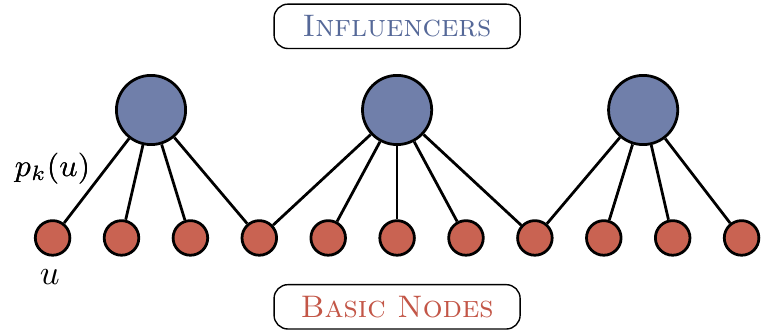}
  \caption{Three influencers with associated activation probabilities $p_k(u)$.}
  \label{fig:depth1}
\end{figure}

We complete the formal setting by assuming the existence of $K$ sets $A_k
\subseteq V$ of basic nodes such that each influencer $k \in [K]$ is connected to
each node in $A_k$. We denote $p_k(u)$ the probability for influencer $k$ to
activate the child node $u \in A_k$.
In this context, the diffusion process can be
abstracted as follows.

\begin{definition}[Influence process]
  When an influencer $k \in [K]$ is selected, each basic node $u \in A_k$ is
  \emph{sampled} for activation, according to its probability $p_k(u)$. The
  \emph{feedback} 
  for $k$'s selection consists of all the
  activated nodes, while the associated \emph{reward} consists only of the
  \emph{newly activated} ones.
\end{definition}

\paragraph*{Remark}
Limiting the influence maximization method to working with a small subset of the node base may  allow to accurately estimate their value more rapidly, even in a highly uncertain environment, hence the algorithmic interest. 
At the same time, this is directly motivated by marketing scenarios involving marketers who may not have knowledge of the entire diffusion graph, only having access to a few
influential people who can diffuse information (the influencers in our
setting), or may simply prefer such a two-step flow of diffusion for various reasons, such as establishing credibility.  Moreover, despite the fact that we model the social reach of every influencer
by 1-hop links to the to-be-influenced nodes, these edges are just an
abstraction of the activation probability, and may represent in reality longer
paths in an underlying unknown real influence graph $G$.


\subsection{Online influencer marketing with persistence}

We are now ready to define the \emph{online influencer marketing with
persistence} task.

\begin{problem}[OIMP]
\label{def:problem}
  Given a set of influencers $[K] := \{1, \ldots, K\}$,
  a \emph{budget} of $N$ trials, and a number $1
  \leq L \leq K$ of influencers to be activated at each trial, the objective of the
  \emph{online influencer marketing with persistence} (OIMP) is to solve the
  following optimization problem:
  \[
    \argmax_{I_n \subseteq [K], |I_n|=L, \forall 1\leqslant n\leqslant N}
    \mathbb{E}\left|\bigcup_{1\leqslant n \leqslant N} S(I_n) \right|.
  \]
\end{problem}

As noticed in~\cite{lei15}, the offline influence maximization can be seen as a
special instance of the online one, where the budget is $N=1$. Note that, in
contrast to persistence-free online influence maximization --~considered, e.g.,
in \cite{vaswani15, wen16}~-- the performance criterion used in OIMP displays
the so-called \emph{diminishing returns property}: the expected number of nodes
activated by successive selections of a given seed is decreasing, due to the
fact that nodes that have already been activated are discounted. We refer to the
expected number of nodes remaining to be activated as the \emph{remaining
potential} of a seed. The diminishing returns property implies that there is no
static best set of seeds to be selected, but that the algorithm must follow an
adaptive policy, which can detect that the remaining potential of a seed is
small and switch to another seed that has been less exploited. Our solution to
this problem has to overcome challenges on two fronts: (1) it needs to estimate
the potential of nodes at each round, without knowing the diffusion model nor
the activation probabilities, and (2) it needs to identify the currently best
influencers, according to their estimated potentials.

Other approaches for the online influence maximization problem rely on
estimating diffusion parameters~\cite{lei15, vaswani15, wen16} -- generally, a
distribution over the influence probability of each edge in the graph. However,
the assumption that one can estimate accurately the diffusion parameters -- and
notably the diffusion probabilities -- may be overly ambitious, especially in
cases where the number of allowed trials (the budget) is rather limited. A
limited trial setting is arguably more in line with real-world campaigns: take
as example political or marketing campaigns, which only last for a few weeks.

In our approach, we work with parameters on \emph{nodes}, instead of edges.
More specifically, these parameters represent the potentials of remaining spreads
from each of the influencer nodes. We stress that this potential can evolve
as the campaign proceeds. In this way, we can go around the dependencies on
specific diffusion models, and furthermore, we can remove entirely the
dependency on a detailed graph topology.

%% file: algorithm.tex

In this section, we describe our UCB-like algorithm, which relies on the
Good-Turing estimator to sequentially select the seeds to activate at
each round, from the available influencers.

\subsection{Remaining potential and Good-Turing estimator}

A good algorithm for
OIMP should aim  at selecting the influencer $k$ with the largest potential for
influencing its children $A_k$. However, the true potential value of an influencer
is \emph{a priori} unknown to the decision maker.


In the following, we index trials by $t$ when referring to the time
of the algorithm, and we index trials by $n$ when referring to the number of
selections of the influencer. For example, the $t$-th spread initiated by the
algorithm is noted $S(t)$ whereas the $n$-th spread of influencer $k$ is noted
$S_{k,n}$.

\begin{definition}[Remaining potential $R_{k}(t)$\label{def:mm}]
  Consider an influencer $k \in [K]$ connected to $A_k$ basic nodes. Let $S(1),
  \ldots, S(t)$ be the set of nodes that were activated during the first $t$
  trials by the seeded influencers. The \emph{remaining potential} $R_k(t)$ is
  the expected number of \emph{new} nodes that would be activated upon starting
  the $t+1$-th cascade from $k$:
  \[
    R_{k}(t) := \sum_{u \in A_k} \mathds{1}\left\{u \notin
               \bigcup_{i=1}^{t} S(i)\right\} p_k(u),
  \]
  where $\mathds{1}\{\cdot\}$ denotes the indicator function.
\end{definition}

Definition~\ref{def:mm} provides a formal way to obtain the remaining potential
of an influencer $k$ at a given time. The optimal policy would simply select the
influencer with the largest remaining potential at each time step.  The
difficulty is, however, that the probabilities $p_k(u)$ are unknown. Hence, we
have to design a \emph{remaining potential estimator} $\hat{R}_{k}(t)$ instead.
It is important to stress that the remaining potential is a random quantity,
because of the dependency on the spreads $S(1), \dots, S(t)$.
Furthermore, due to the diminishing returns property, the sequence
$(S_{k,n})_{n\geq 1}$ is stochastically decreasing.

Following ideas from~\cite{good53, bubeck13}, we now introduce a version of the
Good-Turing statistic, tailored to our problem of rapidly estimating the
remaining potential.  Denoting by $n_k(t)$ the number of times influencer $k$
has been selected after $t$ trials, we let $S_{k,1}, \ldots, S_{k,n_k(t)}$ be the
$n_k(t)$ cascades sampled independently from influencer $k$. We denote by
$U_k(u, t)$ the binary function whose value is $1$ if node $u$ has been
activated \emph{exactly} once by influencer $k$ -- such occurrences are called
\emph{hapaxes} in linguistics -- and $Z_k(u, t)$ the binary function whose value
is $1$ if node $u$ has never been activated by influencer $k$. The idea of the
Good-Turing estimator is to estimate the remaining potential as the proportion
of hapaxes in the $n_k(t)$ sampled cascades, as follows:
\[
  \hat{R}_k(t) := \frac{1}{n_k(t)} \sum_{u \in A_k}U_k(u, t) \prod_{l \neq k} Z_l(u, t).
\]

Albeit simple, this estimator turns out to be quite effective in practice. If an
influencer is connected to a combination of both nodes having high activation
probabilities and nodes having low activation probabilities, then successive
traces sampled from this influencer will result in multiple activations of the
high-probability nodes and few of the low-probability ones. Hence, after
observing a few spreads, the influencer's potential will be low, a fact that
will be captured by the low proportion of hapaxes. In contrast, estimators that
try to estimate each activation probability independently will require a much
larger number of trials to properly estimate the influencer's potential. 

To verify this assumption in reality, we conducted an analysis of the empirical
activation probabilities from a Twitter dataset.  Specifically, we used a
collection of tweets and re-tweets gathered via crawling in August 2012. For
each original tweet, we find all corresponding retweets, and, for each user, we
compute the empirical probability of a retweet occurring -- this, in our case,
is a proxy measure for influence probability.  Specifically, for every user $v$
``influenced'' by $u$, i.e., $v$ retweeted at least one original tweet from $u$
-- we compute the estimated diffusion probability: $p_{u,v} = \left|\text{$u$'s
tweets retweeted by $v$}\right| / \left|\text{tweets by $u$}\right|$. In
Fig.~\ref{fig:histogram} (left), we show the survival function of resulting
empirical probabilities in a log-log plot. We can see that most probabilities
are small -- the 9th decile has value $0.045$.

In Fig.~\ref{fig:histogram} (right), we simulated the activation probabilities
of a set of $50$ nodes whose activation probabilities are chosen randomly from
the Twitter empirical probabilities. Most of the sampled values are low, except
a few relatively high ones. Using this sample as the activation probabilities of
an hypothetical influencer node, we observe on Fig.~\ref{fig:simtwitter} (left)
the cumulative influence spread. The curve first shows a steep increase until
approximately $20$ rounds, where users with high probabilities of conversion
have already been activated, while remaining ones are difficult to activate.

\begin{figure}[t]
  \centering
  \includegraphics[height=4.5cm]{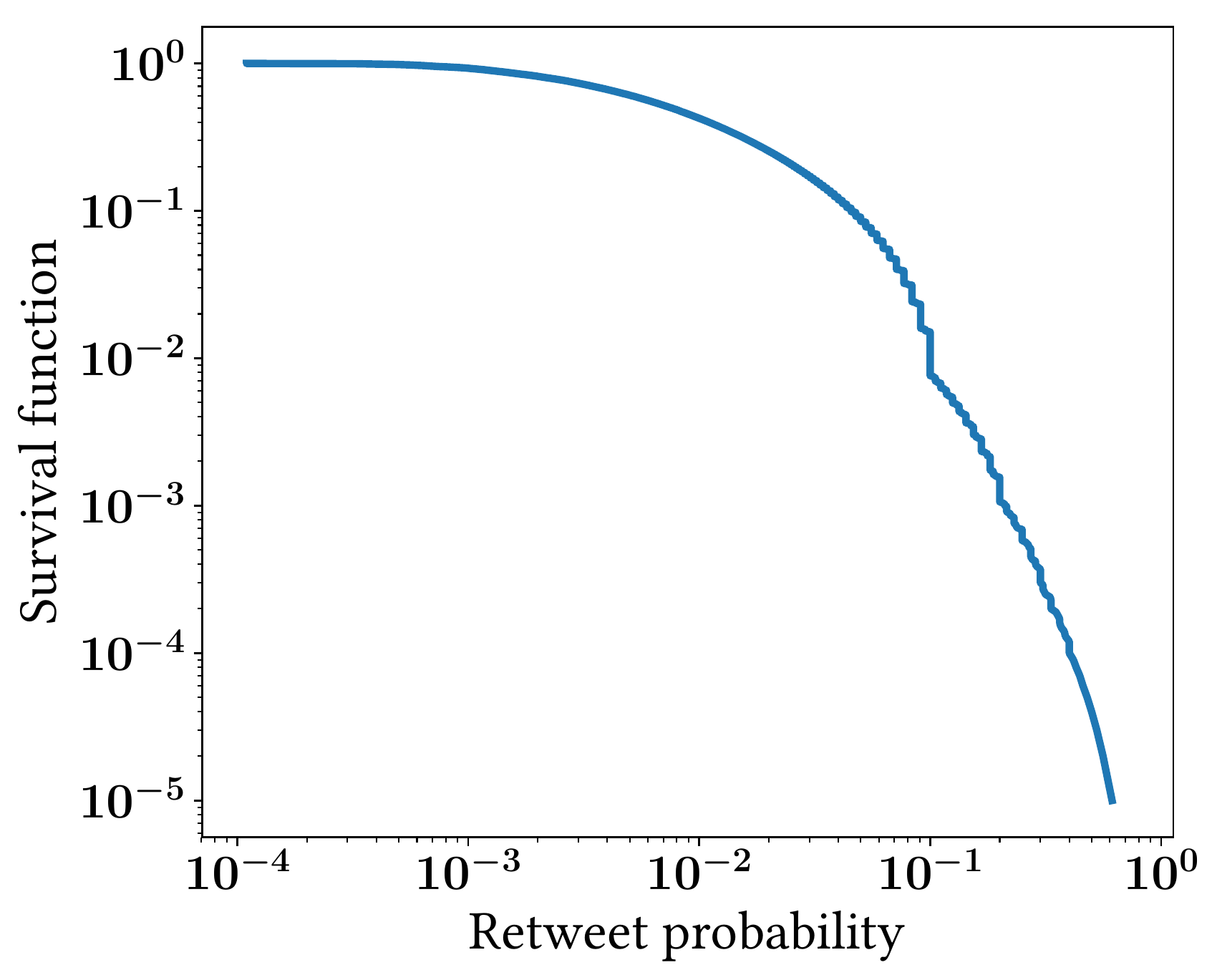}
  \includegraphics[height=4.5cm]{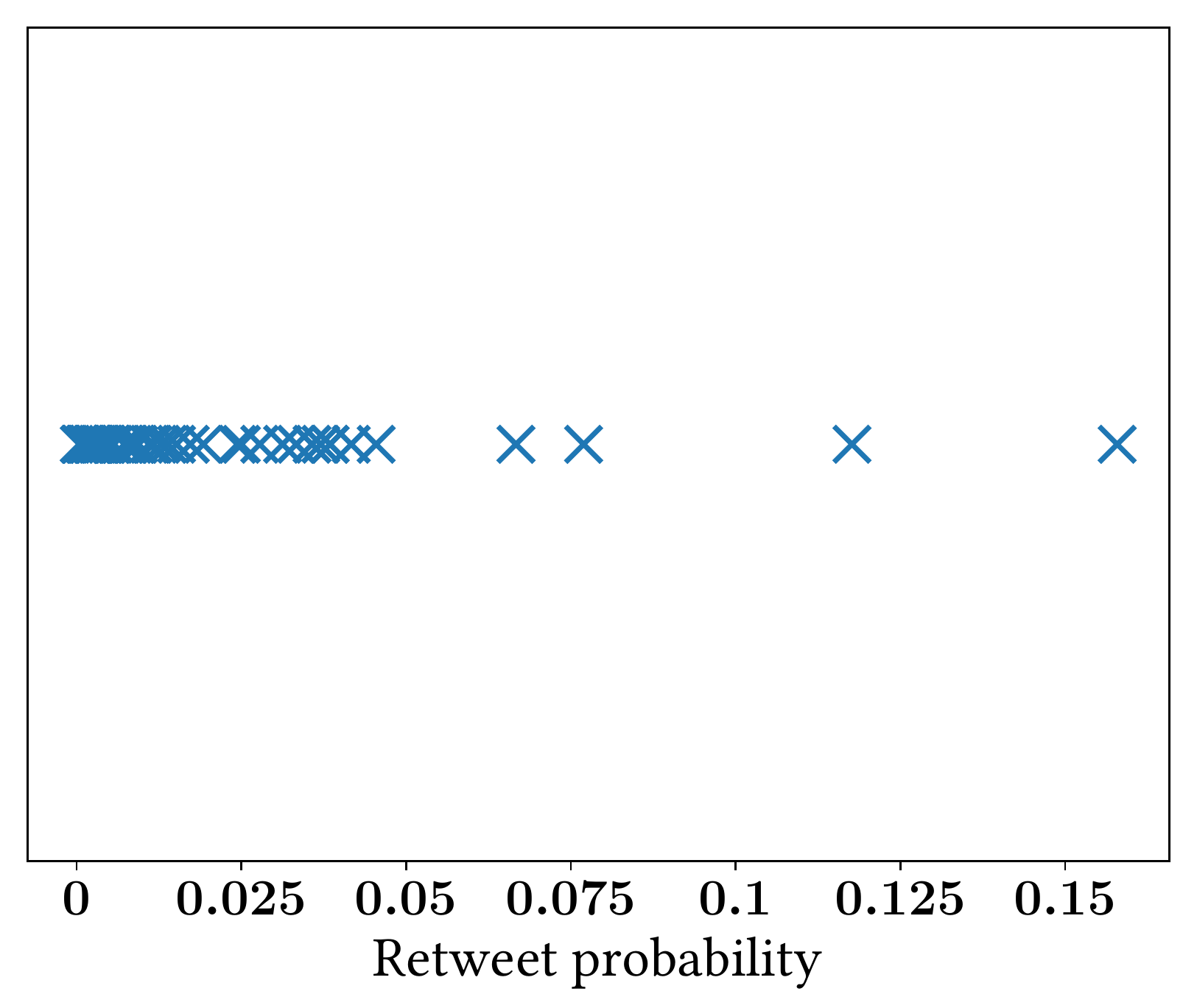}
  \caption{(left) Twitter empirical retweet probabilities. (right) Sample of $50$
  empirical retweet probabilities. \label{fig:histogram}}
\end{figure}

In Fig.~\ref{fig:simtwitter} (right), we compare the Good-Turing estimator to a
Bayesian estimator that maintains a posterior (through a Beta distribution) on
the unknown activation probabilities, updating the posterior after each trial,
similarly to \cite{lei15}. In the Bayesian approach, the remaining potential can
be estimated by summing over the means of the posterior distributions
corresponding to nodes that have not been activated so far. On
Fig.~\ref{fig:simtwitter} (right), the curves are averaged over $200$ runs, and
the shaded regions correspond to the $95\%$ quantiles. Clearly, the Good-Turing
estimator is much faster than its Bayesian counterpart in estimating the actual
remaining potential. Varying the number of nodes --~here equal to 50~-- shows
that the time needed for the Bayesian estimator to provide a reliable estimate
of the remaining potential is proportional to the number of nodes, whereas it
grows only sub-linearly for the Good-Turing estimator.

\begin{figure}[t]
  \centering
  \includegraphics[height=4.5cm]{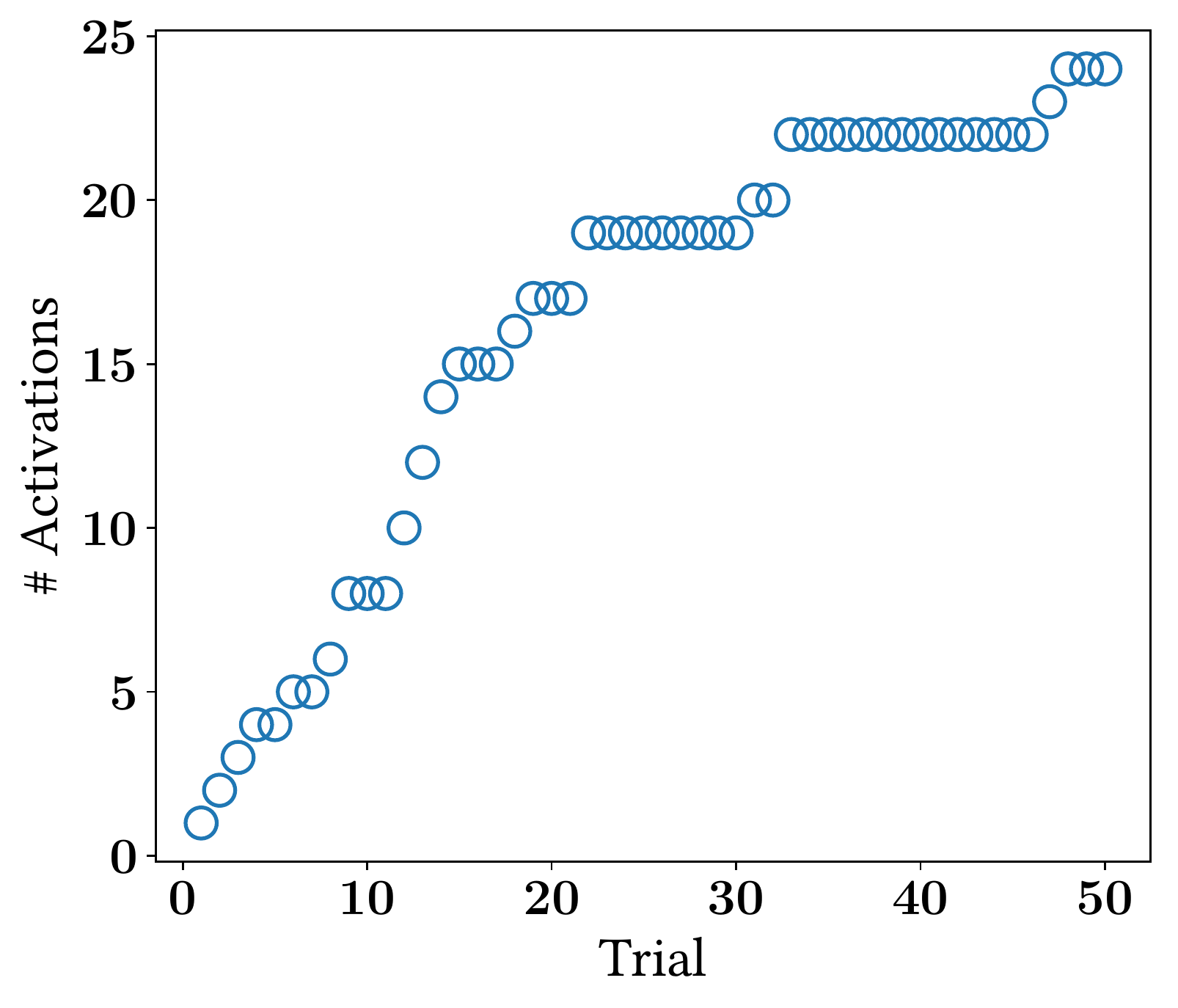}
  \includegraphics[height=4.96cm]{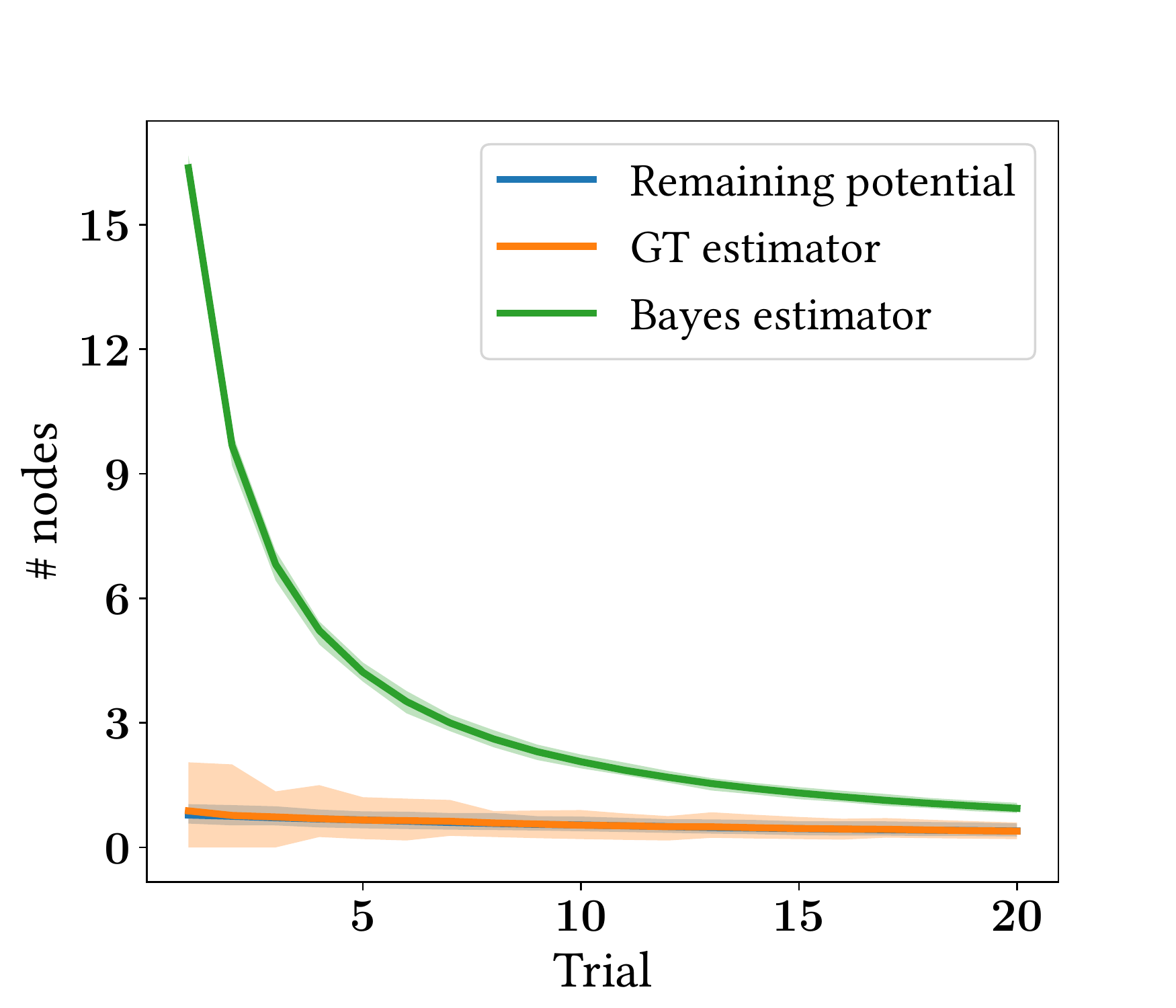}
  \caption{(left) Influence spread against number of rounds. (right) Bayesian
  estimator against Good-Turing estimator.\label{fig:simtwitter}}
\end{figure}

\paragraph*{Remark} While bearing similarities with the traditional missing
mass concept, we highlight one fundamental difference between the remaining potential
and the traditional missing mass studied in~\cite{bubeck13}, which impacts both the
algorithmic solution and the analysis. Since at each step, after selecting an
influencer, \emph{every} node connected to that influencer is sampled, the
algorithm receives a larger feedback than in~\cite{bubeck13}, whose feedback is
in $[0,1]$. However, on the contrary to~\cite{bubeck13}, the hapaxes of an
influencer  $(U_k(u, t))_{u \in A_k}$ are independent. Interestingly, the
quantity $\lambda_k := \sum_{u \in A_k} p(u)$, which corresponds to the expected
number of basic nodes an influencer activates or re-activates in a cascade, will
prove to be a crucial ingredient for our problem.

\subsection{Upper confidence bounds}

Following principles from the bandit literature, the \algoname\ algorithm relies
on \emph{optimism in the face of uncertainty}. At each step (trial) $t$, the
algorithm selects the highest upper-confidence bound on the remaining potential
-- denoted by $b_k(t)$ -- and activates (plays) the corresponding influencer
$k$. This algorithm achieves robustness against the stochastic nature of the
cascades, by ensuring that influencers who ``underperformed'' with respect to
their potential in previous trials may still be selected later on. Consequently,
\algoname\ aims to maintain a degree of \emph{exploration} of influencers,  in
addition to the \emph{exploitation} of the best influencers as per the feedback
gathered so far.

\begin{algorithm}
  \caption{ -- \algoname\ ($L = 1$)}
  \begin{algorithmic}[1]\small
    \REQUIRE{Set of influencers $[K]$, time budget $N$
    }
      \STATE{\textbf{Initialization:} play each influencer
        $k\in[K]$ once, observe the spread $S_{k,1}$, set $n_k=1$}
    \STATE{For each $k\in [K]$: update the reward $W=W\cup S_{k,1}$}
    \FOR{$t = K + 1, \ldots, N$}\label{alg:for}
      \STATE Compute $b_k(t)$ for every influencer $k$
      \STATE Choose $k(t) = \argmax_{k \in [K]} b_k(t)$ \label{alg:optimism}
      \STATE Play influencer $k(t)$ and observe spread $S(t)$
      \STATE Update cumulative reward: $W= W \cup S(t)$
      \STATE Update statistics of influencer $k(t)$: $n_{k(t)}(t+1) = n_{k(t)}(t) + 1$ and
      $S_{k,n_k(t)} = S(t)$.
    \ENDFOR   \label{alg:endfor}
  \RETURN $W$
  \end{algorithmic}
  \label{alg:gooducb}
\end{algorithm}

Algorithm~\ref{alg:gooducb} presents the main components of \algoname\ for the
case $L=1$, that is, when a single influencer is chosen at each step. 

The algorithm starts by activating each influencer $k\in[K]$ once, in order to
initialize its Good-Turing estimator. The main loop  of \algoname\ occurs at
lines \ref{alg:for}-\ref{alg:endfor}. Let $S(t)$ be the observed spread  at
trial $t$, and let $S_{k,s}$ be the result of the $s$-th diffusion initiated at
influencer $k$.  At every step $t > K$, we recompute for each influencer $k \in
[K]$ its index $b_k(t)$, representing the upper confidence bound on the expected
reward in the next trial. The computation of this index uses the previous
samples $S_{k,1},\ldots,S_{k,n_k(t)}$ and the number of times each influencer
$k$ has been activated up to trial $t$, $n_k(t)$. Based on the result of
Theorem~\ref{th:confidence_bounds} --~whose statement and proof are delayed to
Section~\ref{sec:analysis}~--, the upper confidence bound is set as:
\begin{align}\label{eq:ucb}
  b_k(t) = \hat{R}_k(t) + \left(1+\sqrt{2}\right)\sqrt{\frac{\hat{\lambda}_k(t)
  \log(4t)}{n_k(t)}} + \frac{\log(4t)}{3n_k(t)},
\end{align}
where $\hat{R}_k(t)$ is the Good-Turing estimator and $\hat{\lambda}_k(t) :=
\sum_{s=1}^{n_k(t)} \frac{|S_{k,s}|}{n_k(t)}$ is an estimator for the expected
spread from influencer $k$.

Then, in line~\ref{alg:optimism}, \algoname\ selects the influencer $k(t)$ with
the largest index, and initiates a cascade from this node. The feedback $S(t)$
is observed and is used to update the cumulative reward set $W$. We stress again that
$S(t)$ provides only the Ids of the nodes that were activated, with no
information on \emph{how} this diffusion happened in the hidden diffusion
medium.  Finally, the statistics associated to the chosen influencer $k(t)$ are
updated.

\subsection{Extensions for the case $L>1$}

Algorithm~\ref{alg:gooducb} can be easily adapted to select $L > 1$ influencers at
each round. Instead of choosing the influencer maximizing the Good-Turing UCB in
line~\ref{alg:optimism}, we can select those having  the $L$ largest indices.
Note that $k(t)$ then becomes a \emph{set} of $L$ influencers.  A diffusion is
initiated from the associated nodes and, at termination, all activations are
observed.  Similarly to~\cite{vaswani17}, the algorithm requires feedback to
include the influencer responsible for the activation of each node, in order to
update the corresponding statistics accordingly.


%% file: analysis.tex

In this section, we justify the upper confidence bound used by \algoname\ in
Eq.~\ref{eq:ucb} and provide a theoretical analysis of the algorithm.

\subsection{Confidence interval for the remaining potential}

In the following, to simplify the analysis and to allow for a comparison with
the oracle strategy, we assume that the influencers have \emph{non intersecting
support}. This means that each influencer's remaining potential and
corresponding Good-Turing estimator does not dependent on other influencers.
Hence, for notational efficiency, we also omit the subscript denoting the
influencer $k$. After selecting the influencer $n$ times, the Good-Turing
estimator is simply written $\hat{R}_n = \sum_{u \in A} \frac{U_n(u)}{n}$. We
note that the non-intersecting assumption is for theoretical purposes only --
our experiments are done with influencers that can have intersecting supports.

The classic Good-Turing estimator is known to be slightly biased (see Theorem
$1$ in~\cite{mcallester00} for example). We  show in Lemma~\ref{lem:bias} that
our remaining potential estimator adds an additional factor $\lambda = \sum_{u \in A}
p(u)$ to this bias:

\begin{lemma}[]\label{lem:bias}
  The bias of the remaining potential estimator is \[
    \mathbb{E}[R_n] - \mathbb{E}[\hat{R}_n] \in
    \left[-\frac{\lambda}{n},0\right].
  \]
\end{lemma}

\begin{proof}
  \small
  \begin{align*}
      \mathbb{E}[R_n] - &\mathbb{E}[\hat{R}_n] = \sum_{u \in A} \left[p(u)(1 -
        p(u))^n - \frac{n}{n} p(u)(1 - p(u))^{n-1} \right] \\
      &= - \frac{1}{n} \sum_{u \in A} p(u) \times np(u)(1 - p(u))^{n-1} \\
      &= -\frac{1}{n} \mathbb{E}\left[\sum_{u \in A} p(u)U_n(u) \right]
        \in \left[-\frac{\sum_{u\in A}p(u)}{n}, 0\right]\qedhere
  \end{align*}
\end{proof}

Since $\lambda$ is typically very small compared to $|A|$, in expectation, the
estimation should be relatively accurate.  However, in order to understand what
may happen in the worst-case, we need to characterize the deviation of the
Good-Turing estimator:

\begin{theorem}\label{th:confidence_bounds}
  With probability at least $1 - \delta$,  for $\lambda = \sum_{u \in A}
  p(u)$ and $\beta_n := \left(1 + \sqrt{2}\right) \sqrt{\frac{\lambda
  \log(4/\delta)}{n}} + \frac{1}{3n}\log\frac{4}{\delta}$, the following holds:
  \[
    - \beta_n - \frac{\lambda}{n} \leq R_n - \hat{R}_n \leq \beta_n.
  \]
\end{theorem}
Note that the additional term appearing in the left deviation corresponds to the
bias of our estimator, which leads to a non-symmetrical interval.

\begin{proof}
  We prove the confidence interval in three steps:
  \begin{inparaenum}[(1)]
    \item Good-Turing estimator deviation,
    \item remaining potential deviation, 
    \item combination of these two
          inequalities to obtain the final confidence interval.
  \end{inparaenum}

  Here, the child nodes are assumed to be sampled \emph{independently}, which
  is a simplification compared to the classic missing mass concentration
  results that relies on negatively associated samples~\cite{mcallester00,mcallester03}.
  On the other hand, since we may activate several nodes at once, we need original
  concentration arguments to control the increments of both $\hat{R}_n$ and $R_n$.

  \textbf{($1$)~Good-Turing deviations.}
  Let $X_n(u) := \frac{U_n(u)}{n}$. We have that
  \begin{align*}
    v &:= \sum_{u \in A}\mathbb{E}[X_n(u)^2] = \frac{1}{n^2} \sum_{u \in A}
      \mathbb{E}[U_n(u)]
    \leq \frac{\lambda}{n}.
  \end{align*}
  Moreover, clearly the following holds: $X_n(u) \leq \frac{1}{n}$. 
  
  Applying Bennett's inequality (Theorems~2.9,~2.10 in~\cite{boucheron13}) 
  to the independent random variables $\{X_n(u)\}_{u \in A}$ yields
  \begin{align}\label{eq:gtdeviation}
    \mathbb{P}\left(\hat{R}_n - \mathbb{E}[\hat{R}_n] \geq
    \sqrt{\frac{2\lambda\log(1/\delta)}{n}} + \frac{\log(1/\delta)}{3n}\right)
    \leq \delta.
  \end{align}
  The same inequality can be derived for left deviations.

  \textbf{($2$)~Remaining potential deviations.} 
  Remember that $Z_n(u)$ denotes the indicator equal to $1$ if $u$ has never been activated
  up to trial $n$. We can rewrite the remaining potential as $R_n = \sum_{u \in A} Z_n(u)
  p(u).$
  Let  $Y_n(u) = p(u)(Z_n(u) - \mathbb{E}[Z_n(u)])$ and $q(u) =
  \mathbb{P}(Z_n(u) = 1) = (1 - p(u))^n$. For some $t > 0$, we have next that 
  \begin{align*}
    \mathbb{P}(&R_n - \mathbb{E}[R_n] \geq \epsilon) \leq e^{-t \epsilon}
        \prod_{u \in A} \mathbb{E}\left[e^{t Y_n(u)}\right] \\
    &= e^{t\epsilon} \prod_{u\in A} \left(q(u)e^{t p(u)(1 -
        q(u))} + (1 - q(u))e^{-t p(u) q(u)}\right) \\
    &\leq e^{-t\epsilon} \prod_{u\in A} \exp(p(u)t^2/(4n))
    = \exp\left(-t\epsilon + t^2/(4n) \lambda \right).
  \end{align*}
  The first inequality is well-known in exponential concentration bounds and relies
  on Markov's inequality. The second inequality follows from~\cite{berend13} (Lemma~3.5). 
  
  Then, choosing $t = \frac{2n\epsilon}{\lambda}$, we obtain
  \begin{align}\label{eq:mmdeviation}
    \mathbb{P}\left(R_n - \mathbb{E}[R_n] \geq
    \sqrt{\frac{\lambda\log(1/\delta)}{n}}\right) \leq \delta.
  \end{align}
  We can proceed similarly to obtain the left deviation.

  \textbf{(3) Putting it all together.} We combine Lemma~\ref{lem:bias}
   with Eq.~(\ref{eq:gtdeviation}), (\ref{eq:mmdeviation}), to obtain the final
   result. Note that
   $\delta$ is replaced by $\frac{\delta}{4}$ to ensure
   that  both the left and right bounds for the Good-Turing estimator and the
   remaining potential are verified.
\end{proof}

\subsection{Theoretical guarantees}

We now provide an analysis of the \emph{waiting time} (defined below) of
\algoname, by comparing it to the waiting time of an oracle policy, following
ideas from~\cite{bubeck13}. Let $R_k(t)$ be  the remaining potential of influencer $k$ at
trial number $t$. This differs from $R_{k,n}$, which is the remaining potential
of influencer $k$ once \emph{it} has been played $n$ times.

\begin{definition}[Waiting time]
  Let $\lambda_k = \sum_{u\in A_k} p(u)$ denote the expected number of
  activations obtained by the first call to influencer $k$. For $\alpha \in (0,1)$,
  the \emph{waiting time} $T_{UCB}(\alpha)$ of \algoname\ represents the round
  at which the remaining potential of \emph{each} influencer $k$ is smaller than $\alpha
  \lambda_k$. Formally,
  \[
    T_{UCB}(\alpha) := \min \{t : \forall k \in [K], R_k(t) \leq \alpha\lambda_k\}.
  \]
\end{definition}

The above definition can be applied to any strategy for influencer selection and, in
particular, to an oracle one that knows beforehand the $\alpha$ value that is targeted,
the spreads $(S_{k,s})_{k\in[K], 1\leq s \leq t}$ sampled up to the current time,
and the individual activation probabilities $p_k(u), u \in A_k$.
A policy having access to all these aspects will perform the fewest possible activations on each
influencer. We denote by $T^*(\alpha)$ the waiting time of the oracle policy. We are now
ready to state the main theoretical property of the \algoname\ algorithm.

\begin{theorem}[Waiting time]\label{th:waitingtime}
  Let $\lambda^{\text{min}} := \min_{k \in [K]} \lambda_k$ and let
  $\lambda^{\text{max}} := \max_{k \in [K]} \lambda_k$. Assuming  that
  $\lambda^{\text{min}} \geq 13$, for any $\alpha \in
  \left[\frac{13}{\lambda^\text{min}}, 1\right]$, if we define $\tau^* :=
  T^*\left(\alpha - \frac{13}{\lambda^{\text{min}}}\right)$, with probability
  at least $1 - \frac{2K}{\lambda^{\text{max}}}$ the following holds:
  \begin{align*}
    T_{\text{UCB}}(\alpha) \leq \tau^* + K\lambda^{\text{max}} \log(4\tau^* +
    11K\lambda^{\text{max}}) + 2K.
  \end{align*}
\end{theorem}

The proof of this result is given in Appendix~\ref{app:wtanalysis}. 
Unsurprisingly, Theorem~\ref{th:waitingtime} says that \algoname\ must perform
slightly more activations of the influencers than the oracle policy. With high
probability -- assuming that the best influencer has an initial remaining
potential that is much larger than the number of influencers -- the waiting time
of \algoname\ is comparable to $T^*(\alpha')$, up to factor that is only
logarithmic in the waiting time of the oracle strategy. $\alpha'$ is smaller
than $\alpha$ --~hence $T^*(\alpha')$ is larger than $T^*(\alpha)$-- by an
offset that is inversely proportional to the initial remaining potential of the worst
influencer. This essentially says that, if we deal with large graphs, and if the
influencers trigger reasonably large spreads, our algorithm is competitive with
the oracle.

%% file: variant-short.tex

In our study of the OIMP problem so far, a key assumption has been that the influencers
have a \emph{constant} tendency to activate their followers. This hypothesis may not be verified in certain situations, in which the moment influencers promote products
that do not align with their image (\emph{misalignement}) or persist in
promoting the same services (\emph{weariness}). In such cases, they can expect their influence
to diminish~\cite{sletten17,lee17}. To cope with such cases of weariness, we propose in this section an extension to \algoname\ that incorporates the concept of
\emph{influencer fatigue}. 

In terms of bandits, the idea of our extension is similar  in  spirit to the one of Levine et
al.~\cite{levine17}: a new type of bandits  -- called \emph{rotting
bandits}~-- where each arm's value decays as a function of the number of times
it has been selected. We also mention the work of Lou\"edec et
al.~\cite{louedec16} in which the authors propose to take into account the
gradual obsolescence of items to be recommended while allowing new items to be
added to the pool of candidates. In this latter work, the item's value is modeled by a
decreasing function of the number of steps elapsed since the item was added to the pool of
items, whereas in our work --~and in that of~\cite{levine17} --, the value is a
function of the number of times \emph{the item has been selected}.

\subsection{Model adaptation}

The OIMP problem with influencer fatigue can be defined as follows.

\begin{problem}[OIMP with \emph{influencer fatigue}\label{def:problem-fatigue}]
  Given a set of influencers $[K]$, a \emph{budget} of $N$
  trials, a number $1 \leq L \leq K$ of influencers to be activated at each
  trial, the objective of online influencer marketing with persistence
  (OIMP) and with \emph{influencer fatigue} is to solve the following optimization
  problem:
  $$
    \argmax_{I_n \subseteq [K], |I_n|=L, \forall 1\leqslant n\leqslant N}
    \mathbb{E}\left|\bigcup_{1\leqslant n \leqslant N} S(I_n) \right|,
  $$
  knowing that, at the $s$-th selection of an influencer $k \in [K]$, the
  probability that $k$ activates some basic node $u$ is:
  $$
    p_s(u) = \gamma(s) p(u),
  $$
  for $\gamma: \mathbb{N}^* \to (0, 1]$ a \emph{known} non-increasing
  function and $p(u) \in [0,1]$.
\end{problem}

Our initial OIMP formulation can be seen as a special instance of the one with
influencer fatigue, where the non-increasing function $\gamma$ --~referred to as
the weariness function in the following~-- is the constant function $n \mapsto
1$. We  follow the same strategy to solve this new OIMP variant,  by estimating
the remaining potential of a given influencer by an adaptation of the
Good-Turing estimator.  What makes the problem more complex in this setting is
the fact that our hapax statistics must now take into account the round at which they occured.

\subsection{The \algonamefat~algorithm}

As we did previously, to simplify the analysis, we assume that the influencers
have \emph{non intersecting support}.
We redefine the remaining potential in the setting with influencer fatigue as
$$
  R_k(t) := \sum_{u \in A_k} \mathds{1}\{u \text{ never
      activated}\} \gamma(n_k(t) + 1) p(u),
$$
where $p(u)$ is the probability that the influencer
activates node $u$, independently of the number of spreads initiated by the
influencer. Again, the remaining potential is equal to the expected number of
additional conversions upon starting the $t+1$-th cascade from $k$. The
Good-Turing estimator adapted to the setting with influencer fatigue is
defined as follows:
$$
  \hat{R}_k(t) = \frac{1}{n_k(t)} \sum_{u \in A_k} U^\gamma_{n_k(t)}(u),
$$
where $U^\gamma_{k,n}(u) := \sum_{1 \leq i \leq n} \mathds{1}\{X_{k,1}(u) = \ldots
= X_{k, i-1}(u) = X_{k, i + 1}(u) = \ldots = X_{k, n}(u) = 0, X_{k, i}(u) = 1\}
\frac{\gamma(n+1)}{\gamma(i)}$. In short, if $i$ is the round at which a hapax
has been activated, we reweight it by the factor $\gamma(n+1) / \gamma(i)$ since
we are interested in its contribution at the $n+1$-th spread initiated by the
influencer. We provide a formal justification to this estimator by computing its
bias in Appendix~\ref{sec:appendixfatigue}.

Following the same strategy and principles from the bandit literature, the
\algonamefat\ adaptation of \algoname\ selects at each step (trial) $t$ the
highest upper-confidence bound on the remaining potential -- denoted by $b_k(t)$
-- and activates (plays) the corresponding influencer $k$. The upper confidence bound can
now be set as follows (the full details can also be found in Appendix~\ref{sec:appendixfatigue} -- see Theorem~\ref{th:confidence_bounds_fatigue}):
\begin{align}\label{eq:fatucb}
  b_k(t) = \hat{R}_k(t) + \left(1+\sqrt{2}\right)\sqrt{\frac{\hat{\lambda}_k(t)
  \log(4t)}{n_k(t)}} + \frac{\log(4t)}{3n_k(t)},
\end{align}
where $\hat{R}_k(t)$ is the Good-Turing estimator and
$$
  \hat{\lambda}_k(t) := \frac{\gamma(n_k(t) + 1)}{n_k(t)} \sum_{s=1}^{n_k(t)}
  \frac{|S_{k,s}|}{\gamma(s)}
$$
is an estimator for the expected spread from influencer $k$.

%% file: experiments.tex

We conducted  experiments on two types of datasets:
\begin{inparaenum}[(i.)]
    \item two graphs, widely-used in the influence maximization literature, and
    \item a crawled dataset from Twitter, consisting of tweets occurring during August 2012.
\end{inparaenum}
All methods are implemented \footnote{The code is available at
\url{https://github.com/smaniu/oim}.} in C++ and simulations are done on an
Ubuntu 16.04 machine with an Intel Xeon 2.4GHz CPU 20 cores and 98GB of RAM.

\subsection{Extracting influencers from graphs}
\label{sec:influencers}

\algoname\ does not make any assumptions about the topology of the nodes under the scope of influencers. Indeed, in many settings it may be more natural
to assume that the set of influencers is given and that the activations at each
trial can be observed, while the topology of the underlying graph $G$ remain
unknown. In other settings, we may start from an existing social network $G$, in
which case we need to extract a set of $K$ representative influencers from it.
Ideally,  we should choose influencers that have little intersection in their
``scopes of influence'' to avoid useless seed selections.  While this may be
interpreted and performed differently, from one application to another, we
discuss next some of the most natural heuristics for selecting influencers which
we use in our experiments.

\textbf{MaxDegree.} This method selects the $K$ nodes with the highest
out-degrees in $G$. Note that by this criterion we may select influencers with
overlapping influence scopes.

\textbf{Greedy MaxCover.} This strategy follows the well-known greedy
approximation algorithm for selecting a cover of the graph $G$. Specifically,
the algorithm executes the following  steps  $K$ times:
\begin{enumerate}
  \item Select the node with highest out-degree
  \item Remove all out-neighbors of the selected node
\end{enumerate}
To limit intersections among influencer scopes even more, nodes reachable by
more than $1$ hops may be removed at step~(2).

\textbf{DivRank~\cite{mei10}.} DivRank is a PageRank-like method relying on
reinforced random walks, with the goal of producing diverse high-ranking nodes,
while maintaining the rich-gets-richer paradigm. We adapted the original DivRank
procedure by inverting the edge directions. In doing so, we get influential
nodes instead of prestigious ones. By selecting  the $K$ highest scoring nodes
as influencers, the diversity is naturally induced by the reinforcement of
random walks. This ensures that the influencers are fairly scattered in the
graph and should  have limited impact on each other.

\textbf{Influence maximization approximated algorithms.} The fourth method we tested  in our
experiments assigns uniformly at random a propagation probability to each edge
of $G$, assuming the IC model. Then, a state-of-the-art influence maximization algorithm -- PMC in
our experiments -- is executed on $G$ to get the set of $K$ influencers having
the highest potential spread.

\subsection{Graph datasets}

Similarly to~\cite{lei15}, we tested our algorithm on
HepPh and DBLP, two publicly available collaboration networks.
HepPh is a citation graph, where a directed edge is established when
an author cited at least one paper of another author.
In DBLP undirected edges are drawn between authors which have collaborated
on at least one indexed paper. The datasets are
summarized in Table~\ref{table:datasets}. We emphasize that we kept the datasets
relatively small to allow for comparison with computation-heavy baselines, even
though \algoname\ easily scales to large data, as will be illustrated in
Section~\ref{sec:twitterexp}.

\begin{table}[h]
  \centering
  \caption{Summary of the datasets.\label{table:datasets}}
  \begin{tabular}{lcccc}
    \toprule
    \textbf{Dataset} & HepPh & DBLP & Twitter \\
    \midrule
    \# of nodes & $34.5K$ & $317K$ & $11.6M$ \\
    \# of edges & $422K$ & $2.1M$ & $38.4M$ \\
    \bottomrule
  \end{tabular}
\end{table}

\textbf{Diffusion models.} In the work closest to ours, Lei et al.~\cite{lei15}
compared their solution on the Weighted Cascade instance of IC, where the
influence probabilities on incoming edges sum up to 1. More precisely, every
edge $(u,v)$ has weight $1 / d_v$ where $d_v$ is the in-degree of node $v$.  In
this experimental study, and to illustrate that our approach is
diffusion-independent, we added two other diffusion scenarios to the set of
experiments.  First, we included the tri-valency model (TV), which associates
randomly a probability from $\{0.1, 0.01,0.001\}$ to every edge and follows the
IC propagation model.  We also conducted experiments under the Linear Threshold
(LT) model, where the edge probabilities are set like in the WC case and where
thresholds on nodes are sampled uniformly from $[0,1]$.

\textbf{Baselines.} We compare \algoname\ to several baselines. \textsc{Random}
chooses a random influencer at each round. \textsc{MaxDegree} selects the node
with the largest degree at each step $i$, where the degree does not include
previously activated nodes. Finally, \textsc{EG} corresponds to the
confidence-bound explore-exploit method with exponentiated gradient update
from~\cite{lei15}; it is the state-of-the-art method for the OIMP problem (code
provided by the authors). We use this last baseline on WC and TV weighted graphs
and tune parameters in accordance to the results of their experiments: Maximum
Likelihood Estimation is adopted for graph update and edge priors are set to
Beta($1,20$). Note that \textsc{EG} learns parameters for the IC model, and
hence is not applicable for LT. These baselines are compared to an
\textsc{Oracle} that knows beforehand the diffusion model together with
probabilities. At each round, it runs an influence maximization approximated algorithm -- PMC for IC
propagation, SSA for LT. Note that previously activated nodes are not counted
when estimating the value of a node with PMC or SSA, thus, making
\textsc{Oracle} an adaptive strategy.

All experiments are done by fixing the trial horizon $N=500$, a setting that is
in line with many  real-world marketing campaigns, which are fairly short and
do not aim to reach the entire population. 

\begin{figure*}[t!]
  \centering
  \subfloat[HepPh (WC -- Impact of $K$)\label{fig:hepphinfluencersWC}]{\includegraphics[width=0.3\textwidth]{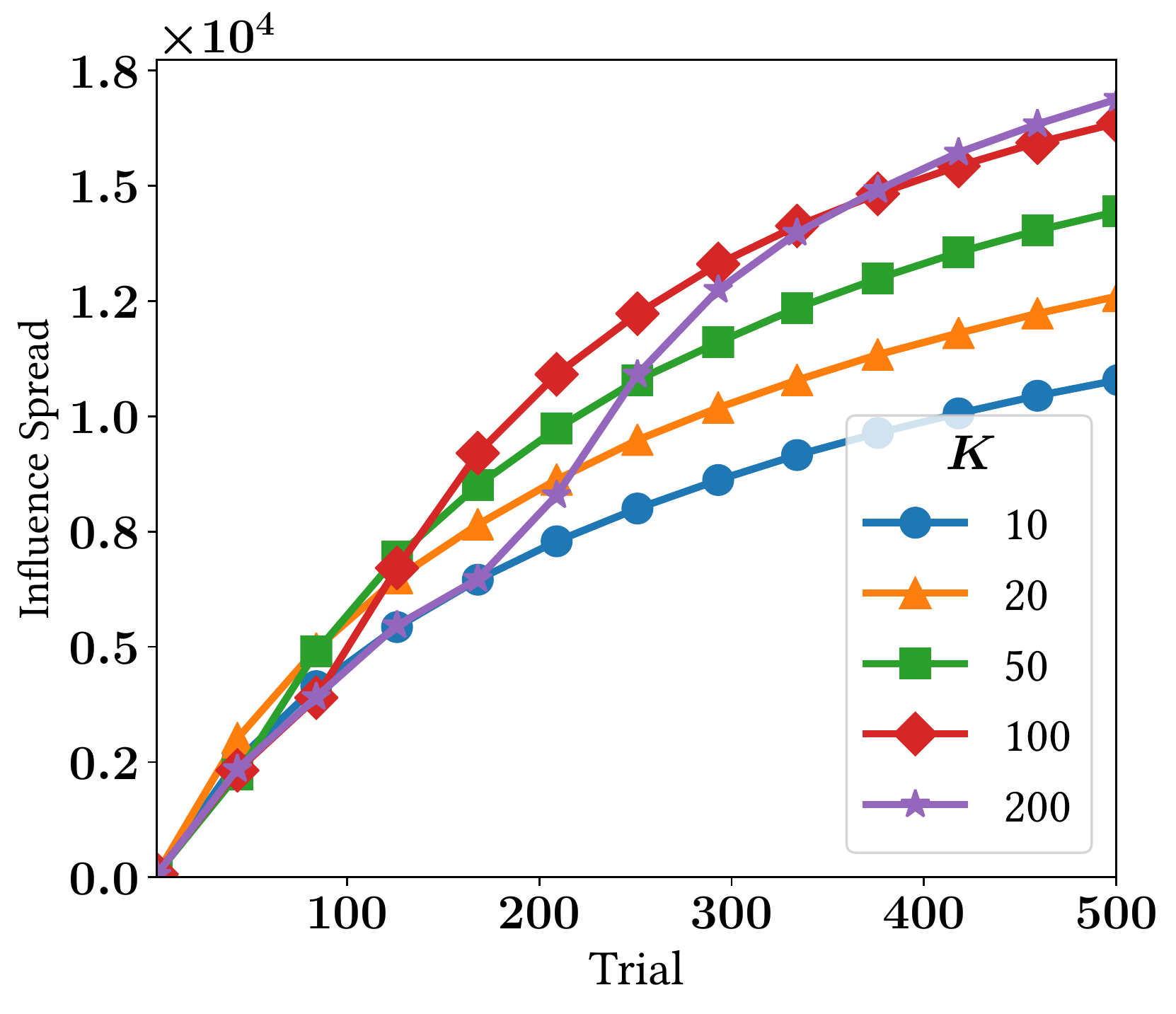}}
  ~~
  \subfloat[HepPh (WC -- Influencer extraction)\label{fig:hepphreductionWC}]{\includegraphics[width=0.3\textwidth]{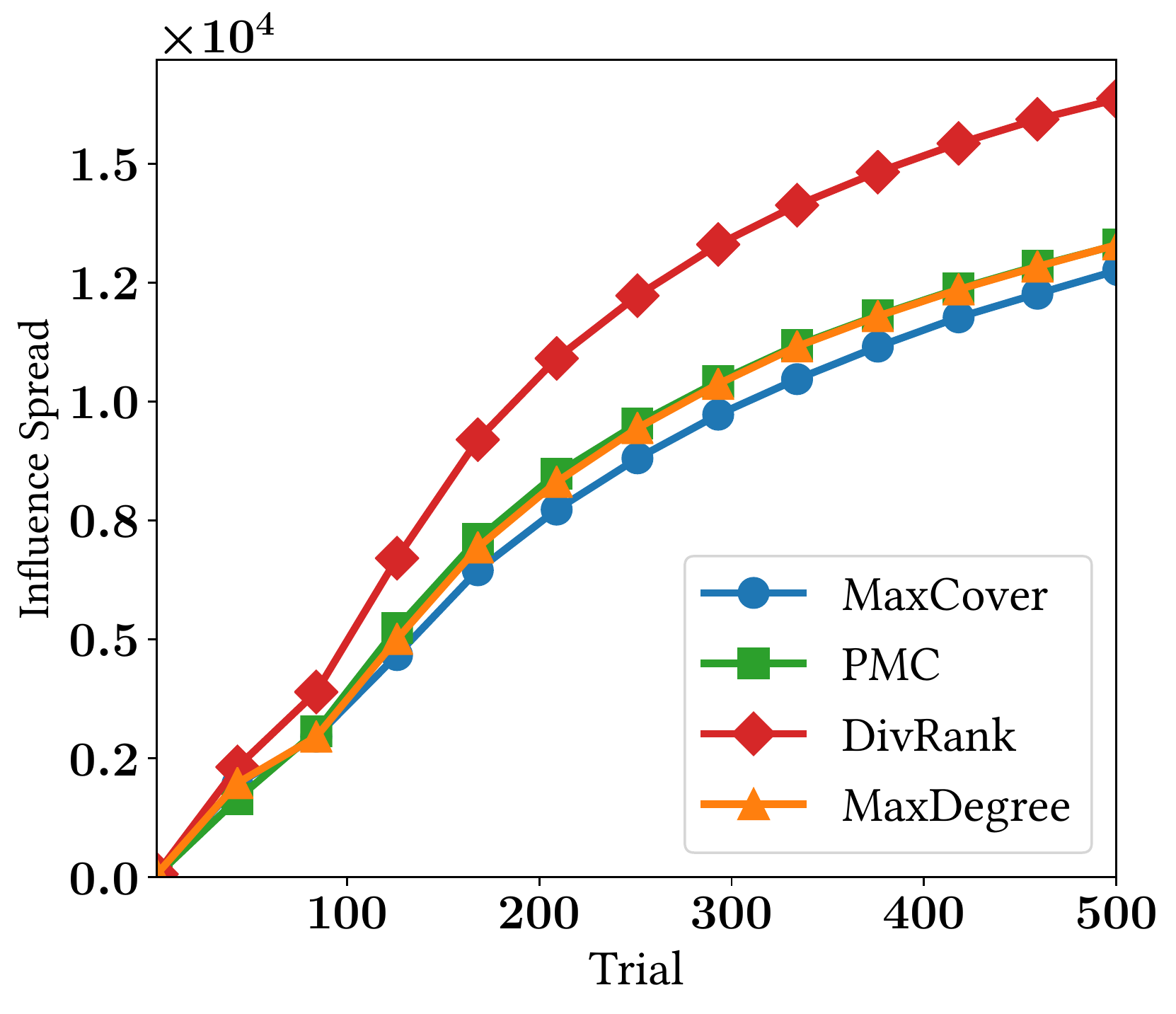}}
  \\
  \subfloat[DBLP (WC -- Impact of $K$)\label{fig:DBLPinfluencersWC}]{\includegraphics[width=0.3\textwidth]{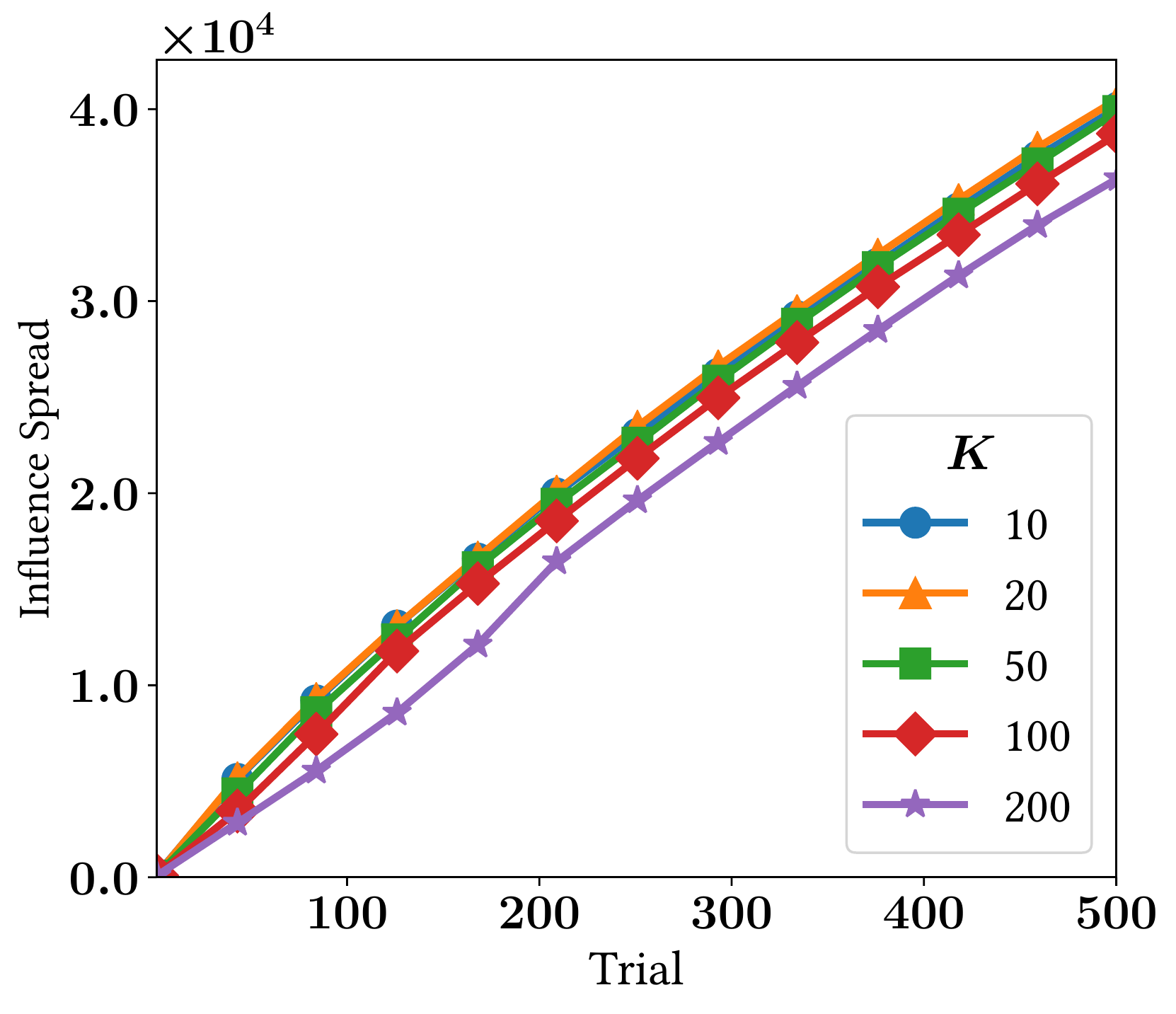}}
  ~~
  \subfloat[DBLP (WC -- Influencer extraction)\label{fig:DBLPreductionWC}]{\includegraphics[width=0.3\textwidth]{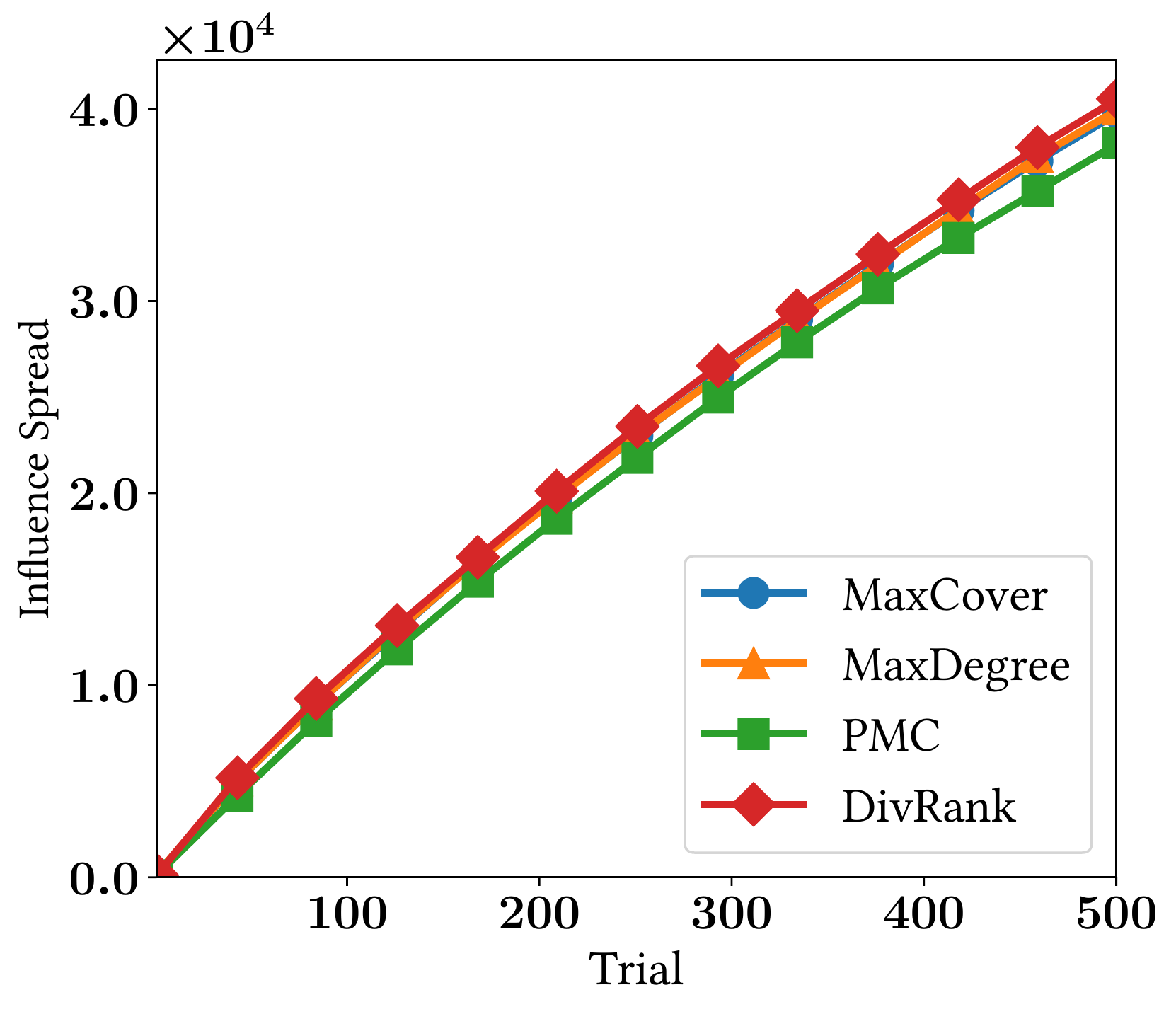}}
  \caption{Impact of $K$ and the influencer extraction criterion on influence spread.\label{fig:hyperparameters}}
\end{figure*}

\textbf{Choice of the influencers.} We show in Fig.~\ref{fig:hepphreductionWC}
and~\ref{fig:DBLPreductionWC} the impact of the influencer extraction criterion
on HepPh and DBLP under WC model. We can observe that the spread is only
slightly affected by the extraction criterion: different datasets lead to
different optimal criteria.  On the HepPh network, DivRank clearly leads to
larger influence spreads. On DBLP, however, the extraction method has little
impact on resulting spreads. We emphasize that on some other graph and model
combinations we observed that other extraction routines can perform better than
DivRank.  In summary, we note that \algoname\ performs consistently as long as
the method leads to influencers that are well spread over the graph. In the
following, for each graph, we used DivRank as the influencer extraction criterion in
accordance with these observations.

In Fig.~\ref{fig:hepphinfluencersWC} and~\ref{fig:DBLPinfluencersWC}, we measure
the impact of the number of influencers $K$ on the influence spread. We can
observe that, on DBLP, a small number of influencers is sufficient to yield
high-quality results. If too many influencers (relative to the budget) are
selected (e.g. $K=200$), the initialization step required by \algoname\ is too
long relative to the full budget, and hence \algoname\ does not reach its
optimal spread -- some influencers still have a large remaining potential at the
end.  On the other hand, a larger amount of influencers leads to greater
influence spreads on HepPh: this network is relatively small ($34.5K$ nodes),
and thus half of the nodes are already activated after $400$ trials. By having
more influencers, we are able to access parts of the network that would not be
accessible otherwise.

\begin{figure}[h]
  \centering
  \includegraphics[width=0.45\textwidth]{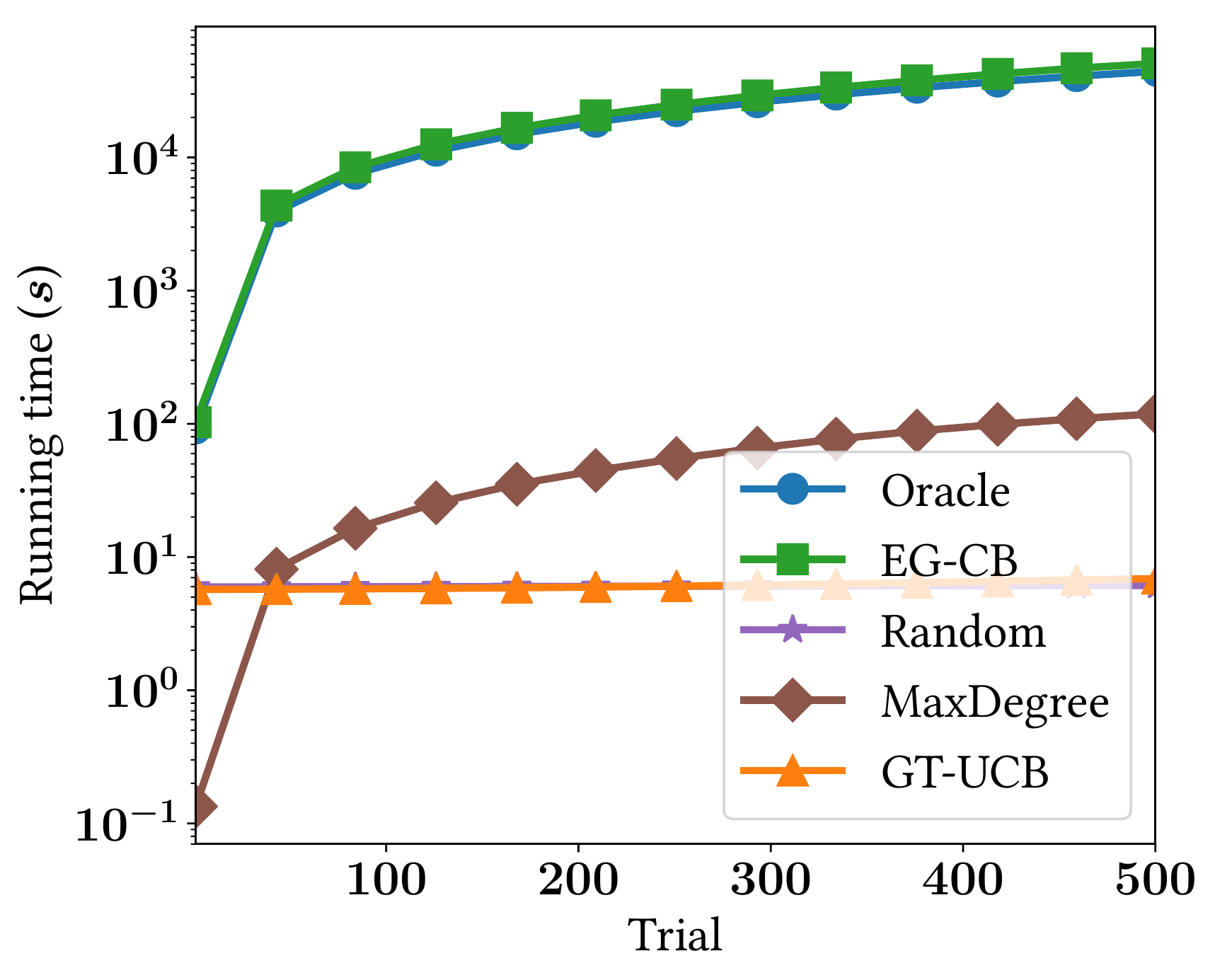}
  \caption{DBLP (WC) -- Execution time.\label{fig:executiontime}}
\end{figure}

\textbf{\algoname\ vs. baselines.} We evaluate the execution time of the
different algorithms in Fig.~\ref{fig:executiontime}. As expected, \algoname\
largely outperforms~\textsc{EG} (and~\textsc{Oracle}). The two baselines require
the execution of an approximated influence maximization algorithm at each round.  In line
with~\cite{arora17}, we observed that SSA has prohibitive computational cost
when incoming edge weights do not sum up to $1$, which is the case with both WC
and TV. Thus, both~\textsc{Oracle} and~\textsc{EG} run PMC on all our
experiments with IC propagation.  \algoname\ is several orders of magnitude
faster: it concentrates most of its running time on extracting influencers, while
statistic updates and UCB computations are negligible.

In Fig.~\ref{fig:baselines}, we show the growth of the spread for \algoname\ and
baselines.  For each experiment, \algoname\ uses $K=50$ if $L=1$ and $K=100$ if
$L = 10$.  First, we can see that~\textsc{MaxDegree} is a strong baseline in
many cases, especially for WC and LT.  \algoname\ results in good quality
spreads across every combination of network and diffusion model. Interestingly,
on the smaller graph HepPh, we observe an increase in the slope of spread after
initialization, particularly visible at $t=50$ with WC and LT. This corresponds
to the step when \algoname\ starts to select influencers maximizing $b_k(t)$ in
the main loop. It shows that our strategy adapts well to the previous
activations, and chooses good influencers at each iteration.  Interestingly,
\textsc{Random} performs surprisingly well in many cases, especially under TV
weight assignment. However, when certain influencers are significantly better
than others, it cannot adapt to select the best influencer unlike \algoname.
\textsc{EG} performs well on HepPh, especially under TV weight assignment.
However, it fails to provide competitive cumulative spreads on DBLP. We believe
that~\textsc{EG} tries to estimate too many parameters for a horizon $T = 500$.
After reaching this time step, less than $10\%$ of all nodes for WC, and $20\%$
for TV, are activated. This implies that we have hardly any information regarding
the majority of edge probabilities, as most nodes are located in parts of the graph
that have never been explored.

\begin{figure*}[t!]
  \centering
  \subfloat[HepPh (WC -- $L=1$)\label{fig:HepphWC}]{\includegraphics[width=0.28\textwidth]{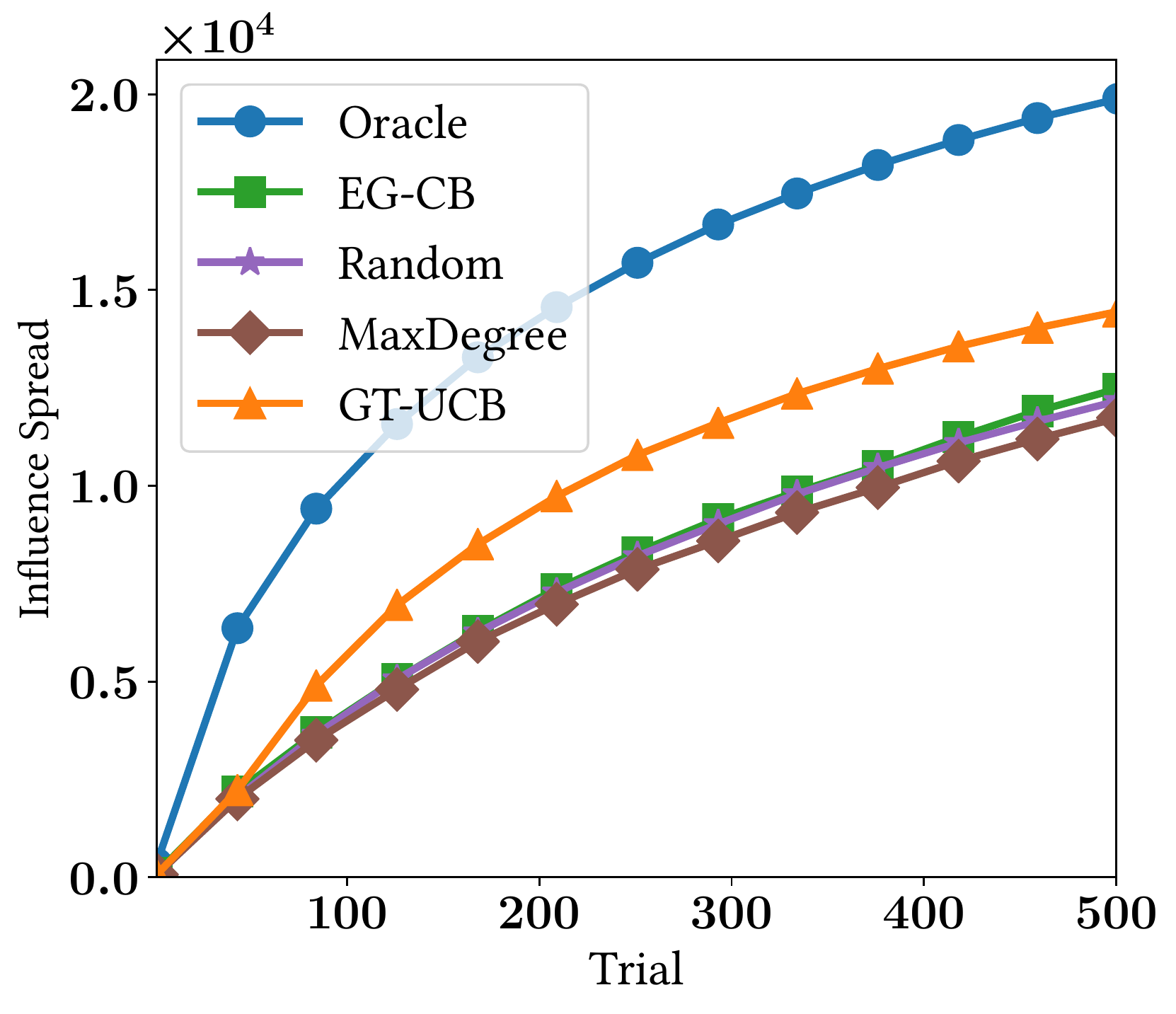}}
  ~
  \subfloat[DBLP (WC -- $L=1$)\label{fig:DBLP1WC}]{\includegraphics[width=0.28\textwidth]{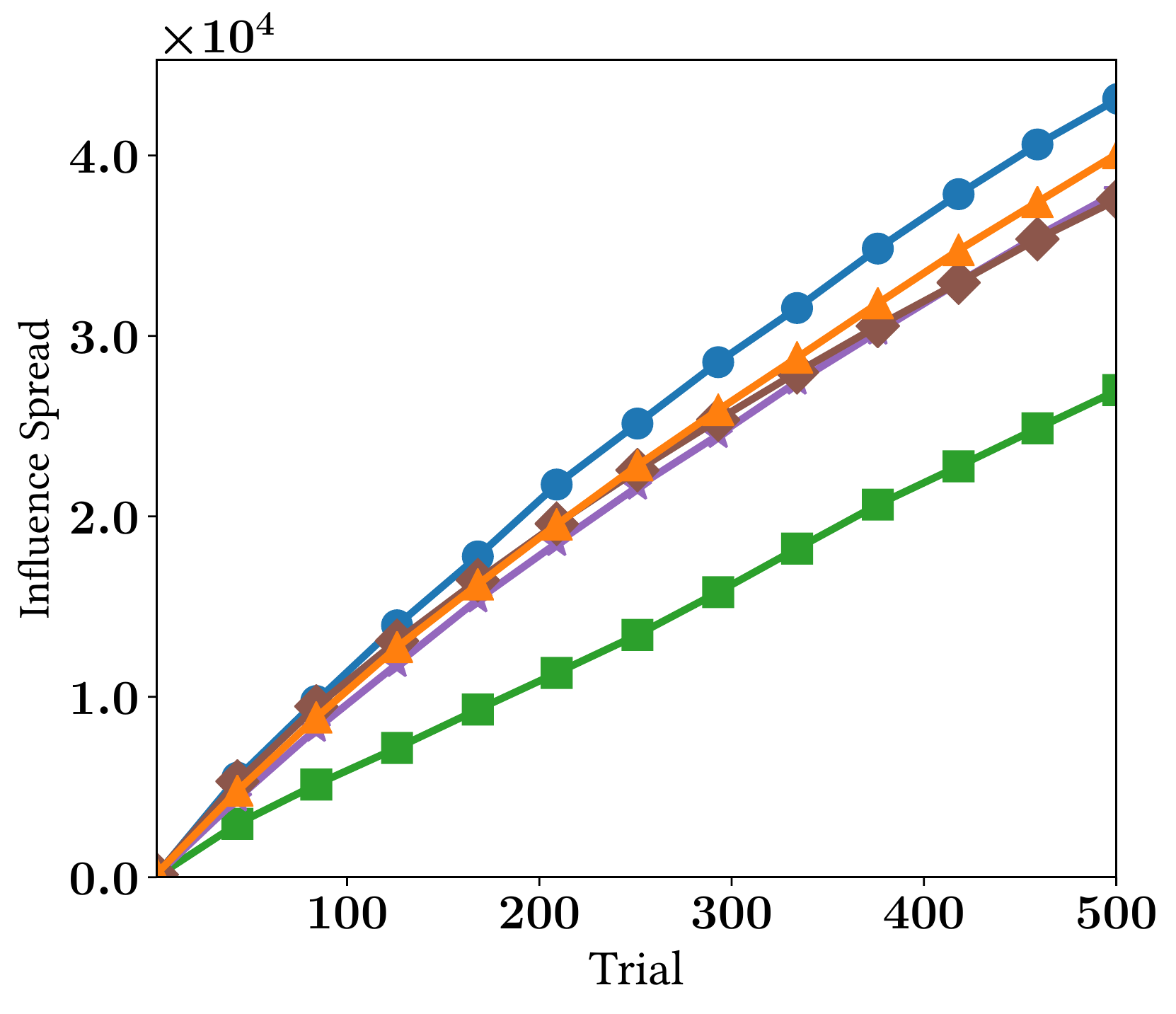}}
  ~
  \subfloat[DBLP (WC -- $L=10$)\label{fig:DBLP5WC}]{\includegraphics[width=0.28\textwidth]{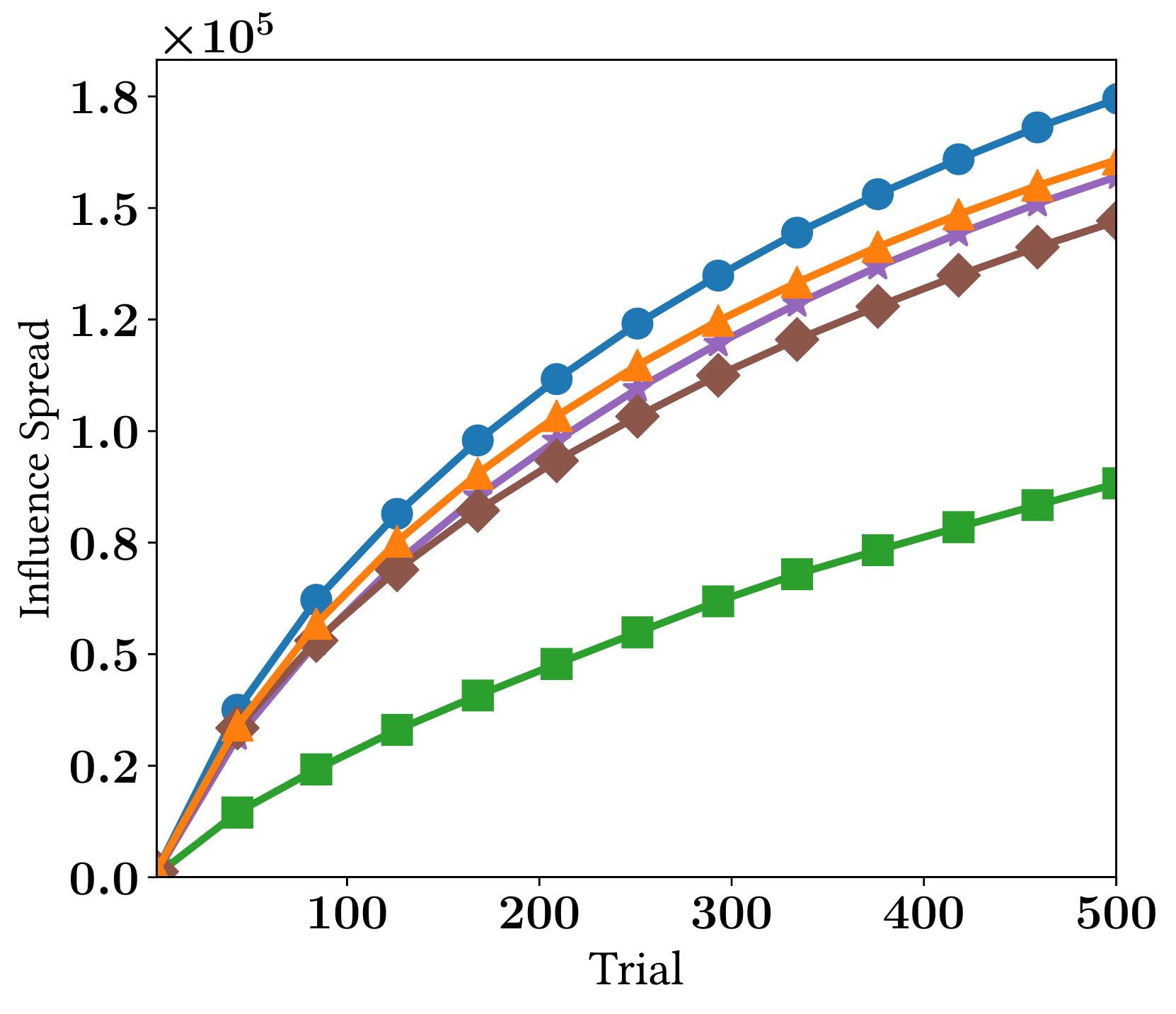}}
  \\
  \subfloat[HepPh (TV -- $L=1$)\label{fig:HepphTV}]{\includegraphics[width=0.28\textwidth]{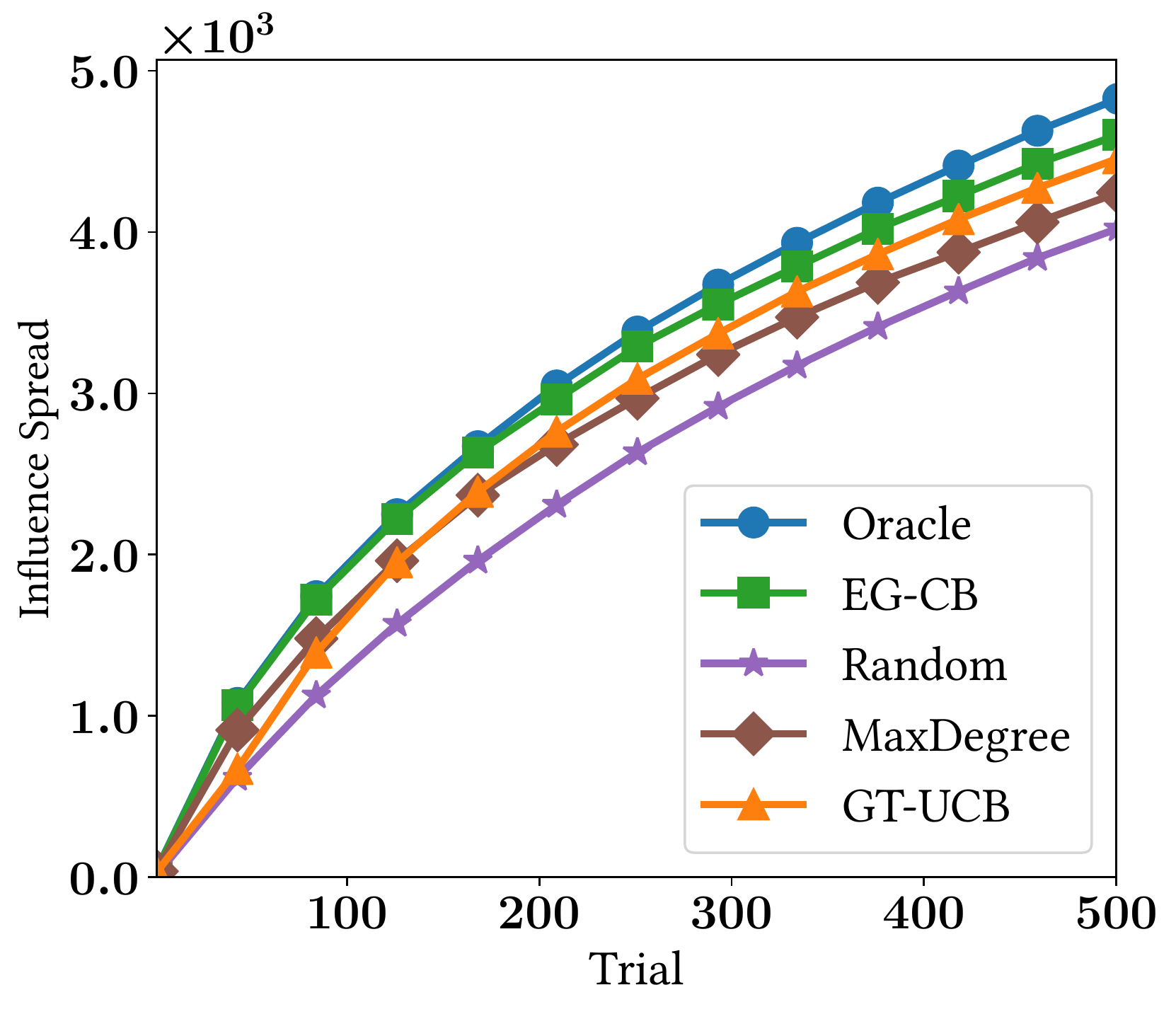}}
  ~
  \subfloat[DBLP (TV -- $L=1$)\label{fig:DBLP1TV}]{\includegraphics[width=0.28\textwidth]{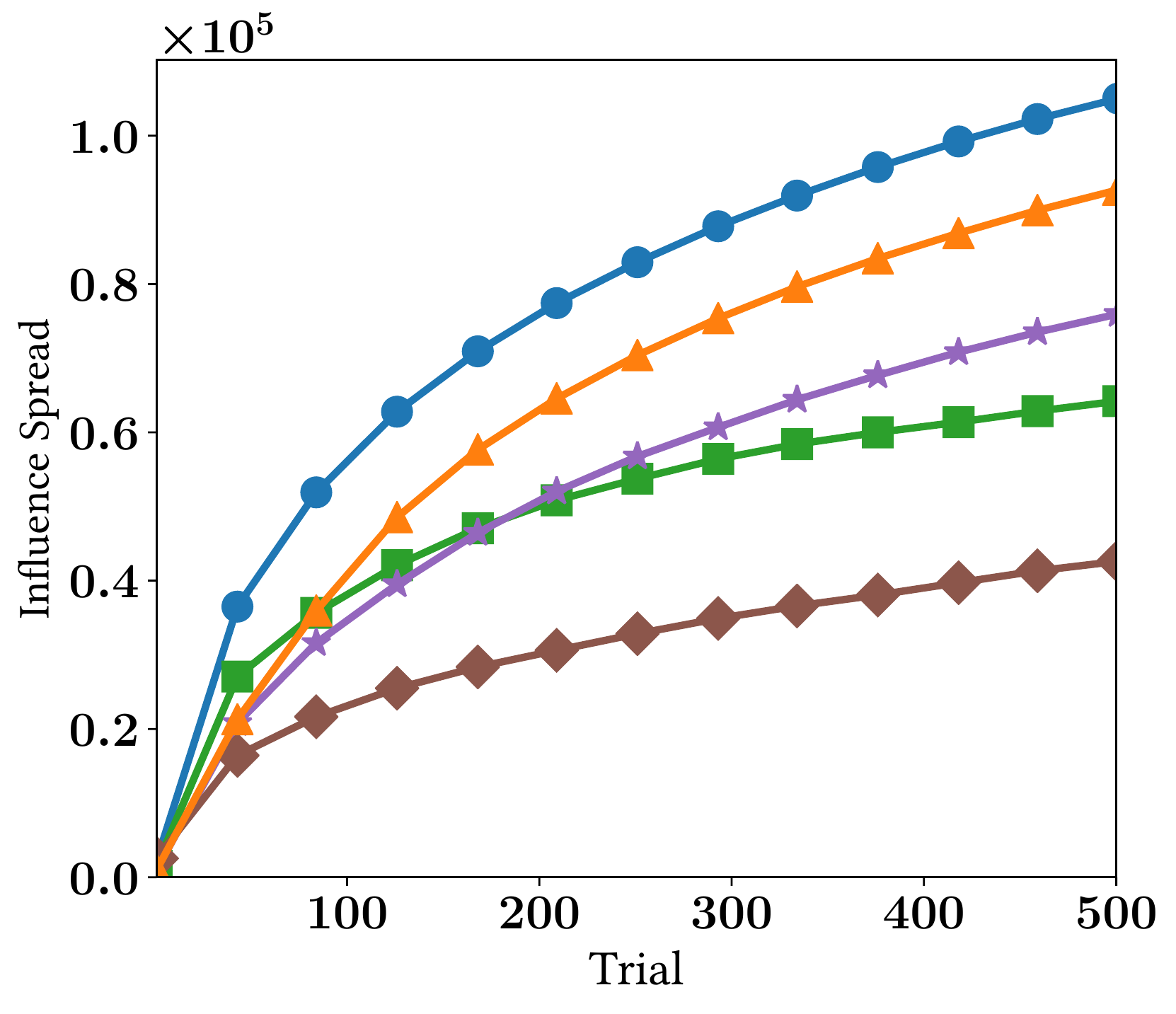}}
  ~
  \subfloat[DBLP (TV -- $L=10$)\label{fig:DBLP5TV}]{\includegraphics[width=0.28\textwidth]{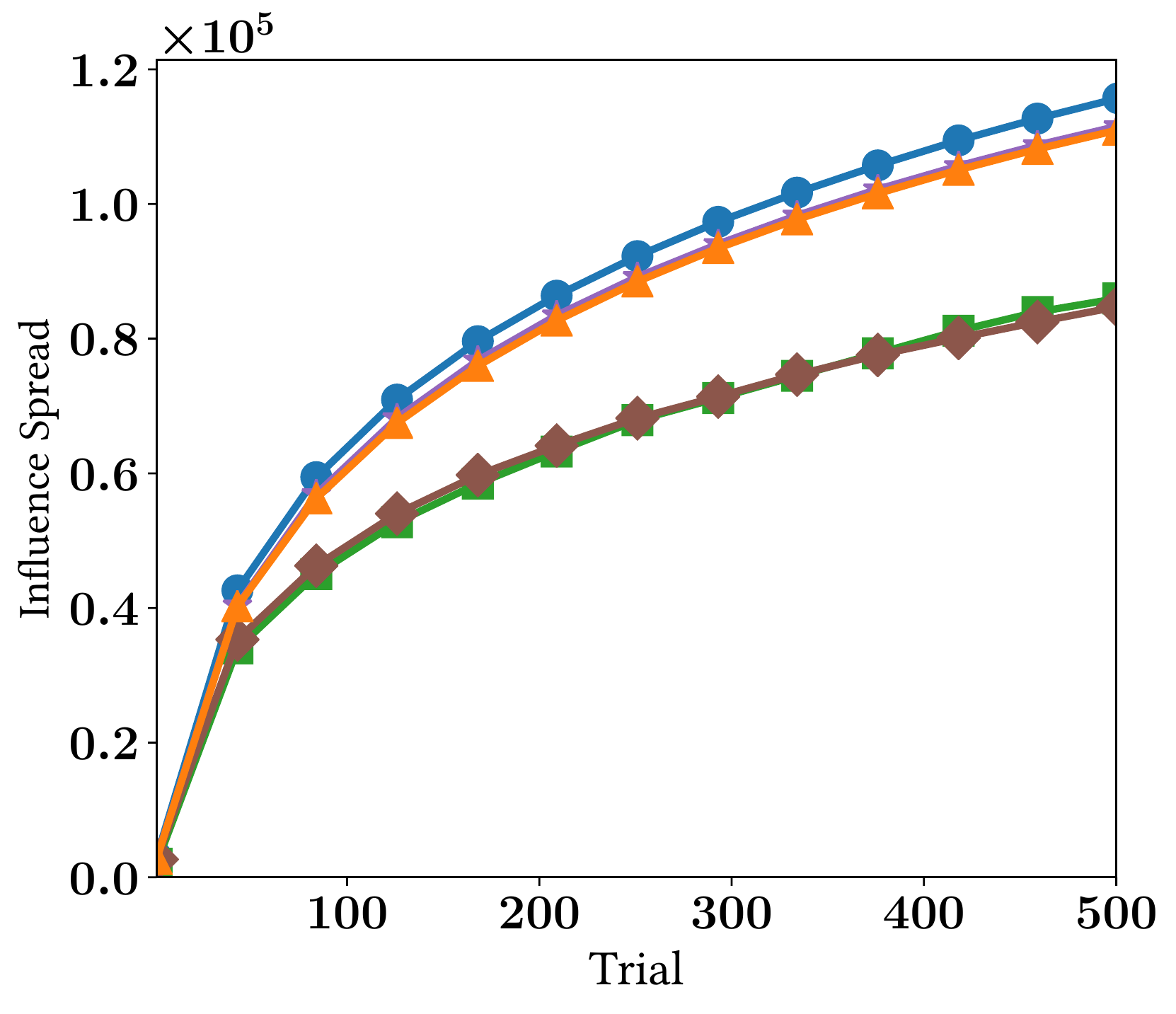}}
  \\
  \subfloat[HepPh (LT -- $L=1$)\label{fig:HepphLT}]{\includegraphics[width=0.28\textwidth]{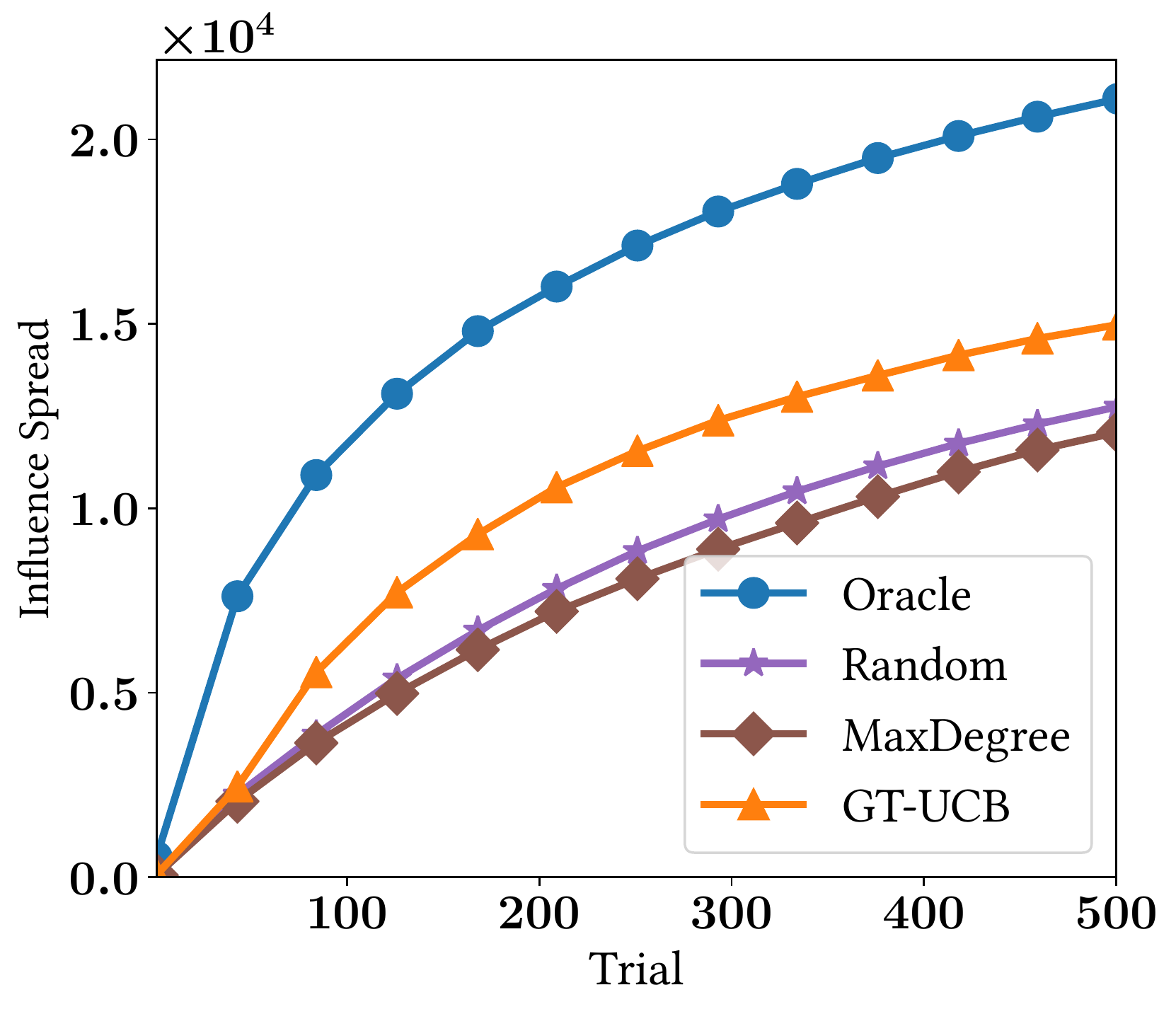}}
  ~
  \subfloat[DBLP (LT -- $L=1$)\label{fig:DBLP1LT}]{\includegraphics[width=0.28\textwidth]{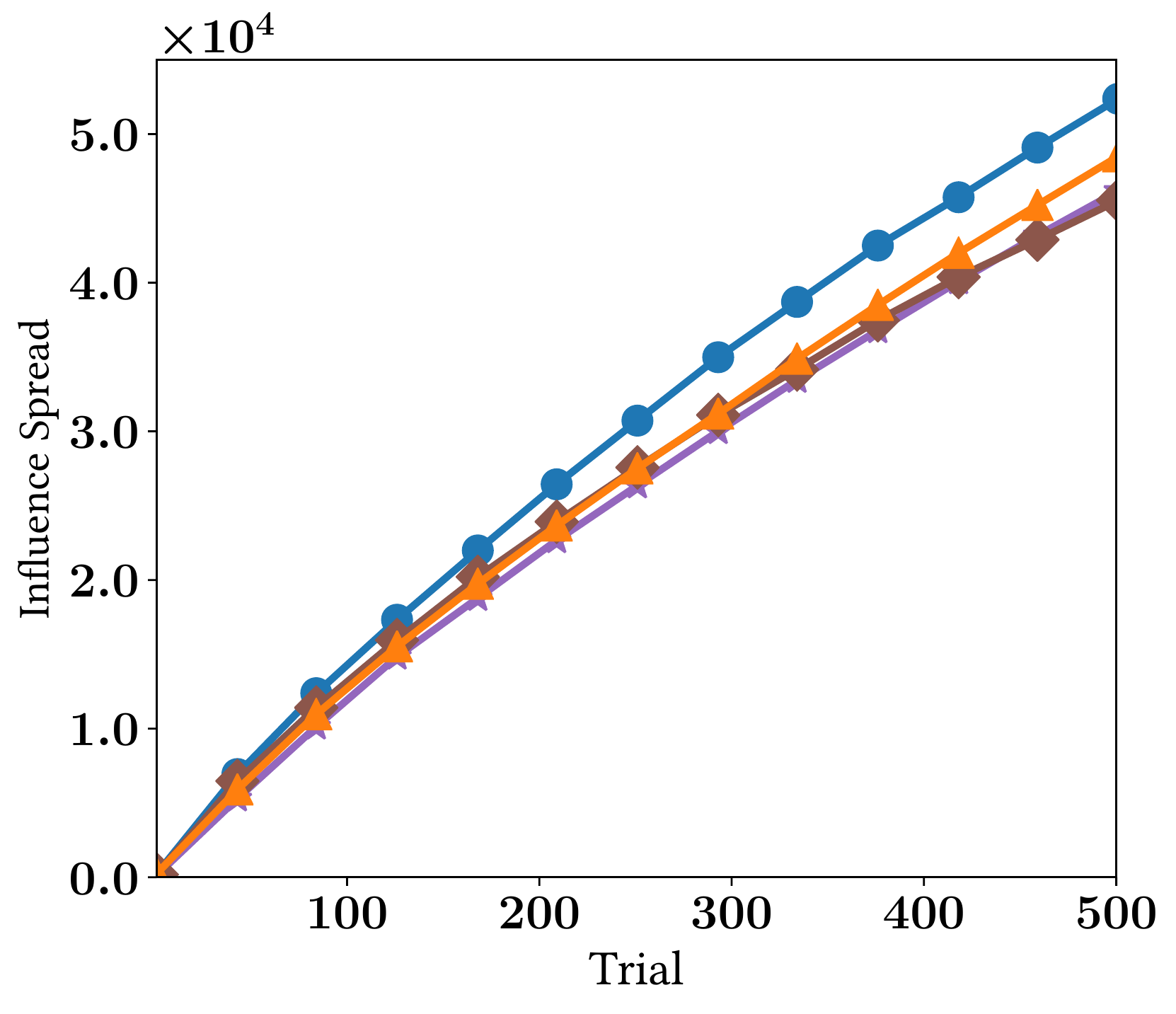}}
  ~
  \subfloat[DBLP (LT -- $L=10$)\label{fig:DBLP5LT}]{\includegraphics[width=0.28\textwidth]{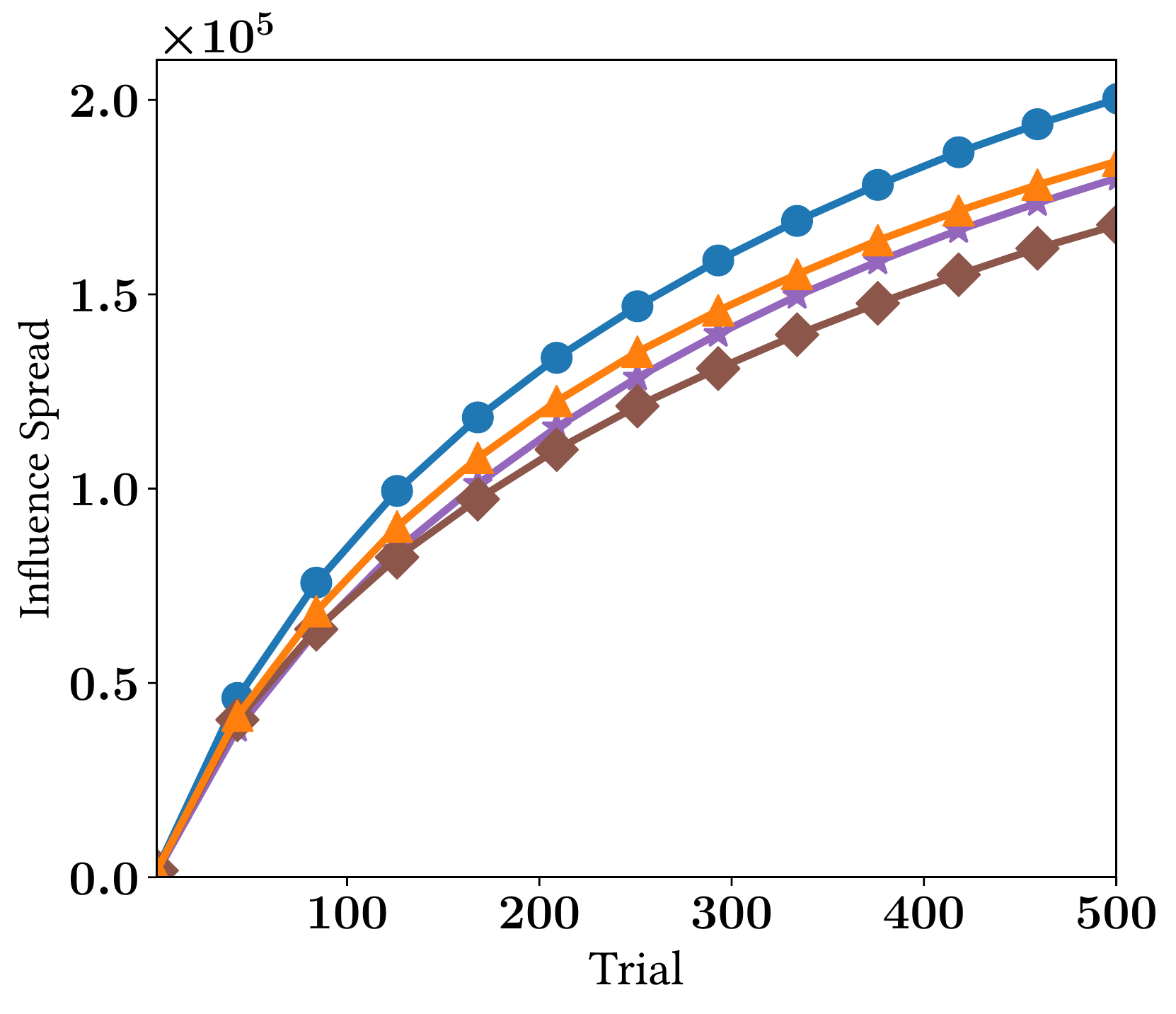}}
  \caption{Growth of spreads against the number of rounds.\label{fig:baselines}}
  \vspace{-2mm}
\end{figure*}

\subsection{Experiments on Twitter}
\label{sec:twitterexp}

We continue the experimental section with an evaluation of \algoname\ on the
Twitter data, introduced as a motivating example in Section~\ref{sec:algorithm}.
The interest of this experiment is to observe actual spreads, instead of
simulated ones, over data that does not  provide an explicit influence graph.

\begin{figure}[h]
  \centering
  \includegraphics[width=0.4\textwidth]{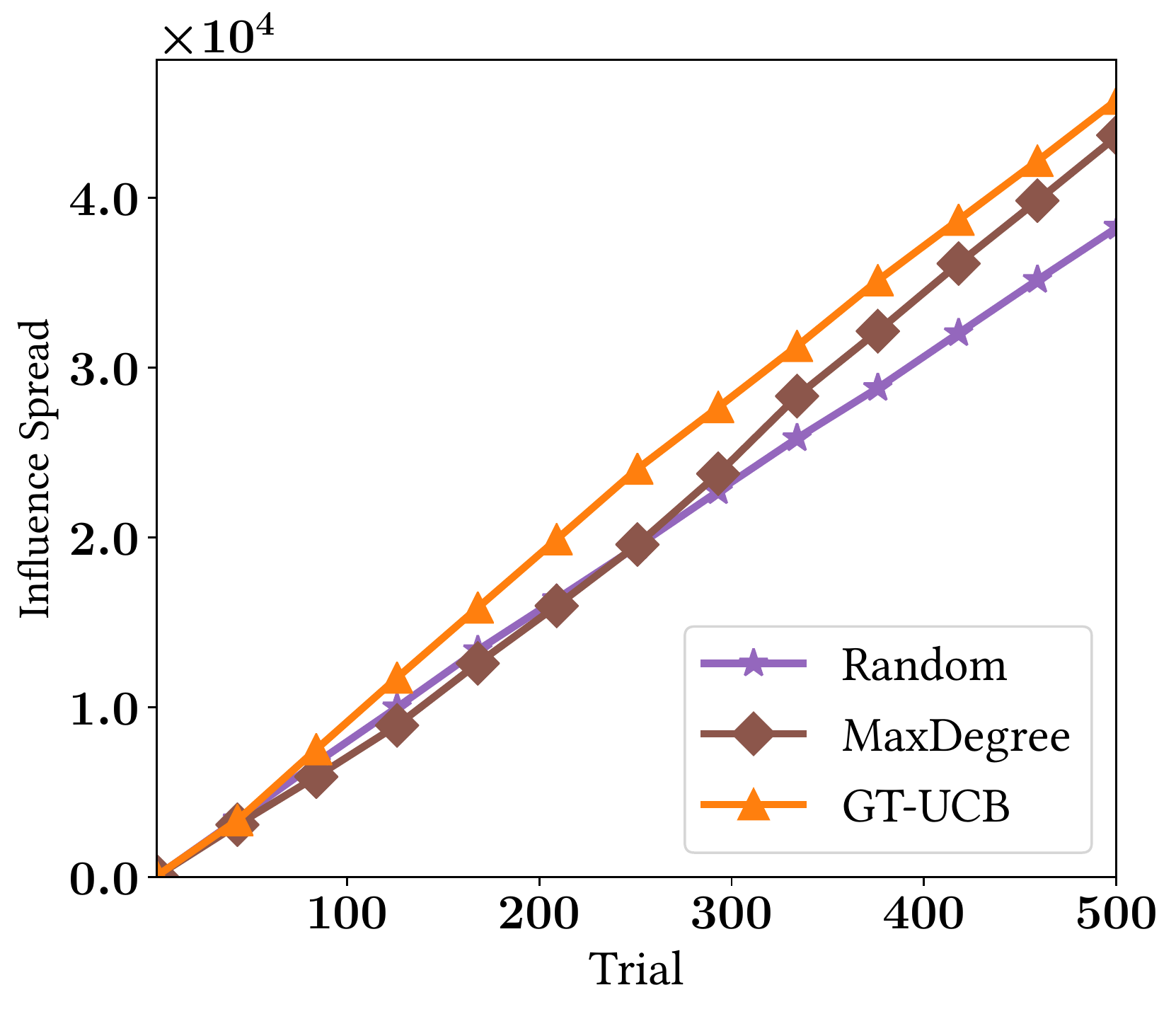}
  \includegraphics[width=0.4\textwidth]{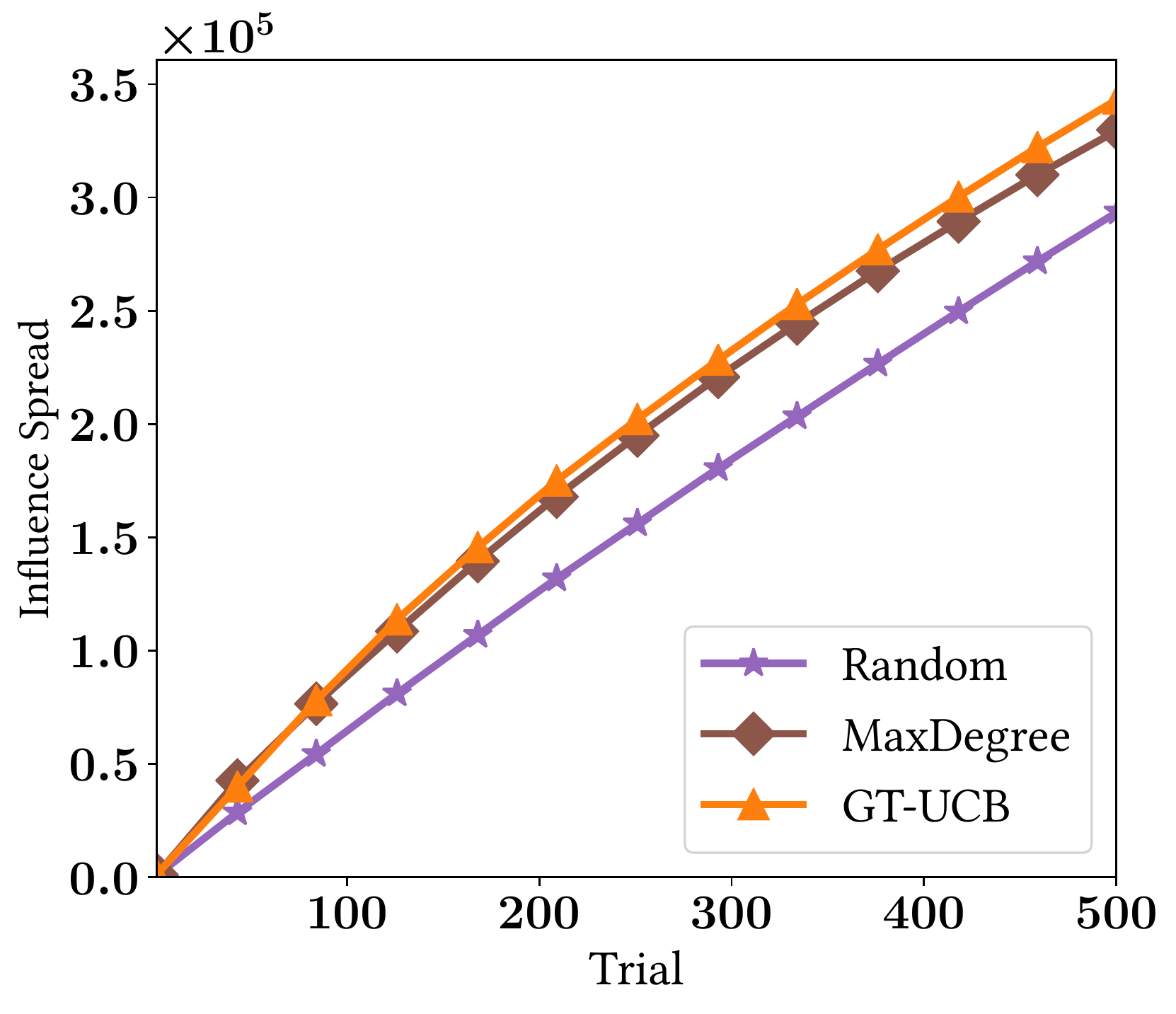}
  \caption{Twitter spread against rounds: (left) $L = 1$ (right) $L = 10$.\label{fig:twitter}}
  \vspace{-2mm}
\end{figure}

From the retweeting logs, for each \emph{active} user $u$ -- a user who posted
more than $10$ tweets -- we select users having retweeted at least one of $u$'s
tweets. By doing so, we obtain the set of potentially influenceable users
associated to active users. We then apply the greedy algorithm to select the
users maximizing the corresponding set cover. These are the influencers of
\algoname\ and \textsc{Random}.  \textsc{MaxDegree} is given the entire
reconstructed network (described in Table~\ref{table:datasets}), that is, the
network connecting active users to re-tweeters.


To test realistic spreads, at each step, once an influencer is selected by
\algoname, a random cascade initiated by that influencer is chosen 
from the logs and we record its spread. This provides realistic, model-free
spread samples to the compared algorithms. Since Twitter only contains
successful activations (re-tweets) and not the failed ones, we could not
test against \textsc{EG}, which needs both kinds of feedback.

In Fig.~\ref{fig:twitter}, we show the growth of the diffusion spread of
\algoname\ against \textsc{MaxDegree} and \textsc{Random}.  Again, \algoname\
uses $K=50$ if $L=1$ and $K=100$ if $L = 10$.  We can see that \algoname\
outperforms the baselines, especially when a single node is selected at each
round. We can observe that \textsc{MaxDegree} performs surprisingly well in both
experiments. We emphasize that it relies on the knowledge of the entire network
reconstructed from retweeting logs, whereas \algoname\ is only given a set of
(few) fixed influencers.


\subsection{Influencer fatigue}

We conclude the experimental section with a series of experiments on Twitter
data and taking into account influencer fatigue.

In a similar way to Section~\ref{sec:twitterexp}, we compute the set of potentially
influenceable users (the support) associated to all active users --~the set of
all users who retweeted at least one tweet from the active user. We then choose
20 influencers as follows: we take the 5 best influencers, that is, the 5 active
users with the largest support; then, the 51st to 55th best influencers, then,
the  501st to 505th best influencers, and finally the 5 worst influencers. By
doing so, we obtain a set of 20 influencers with diverse profiles, roughly covering the possible influencing outcomes. Ideally, a
good algorithm that takes into account influencer fatigue, would need to focus on the 5 best influencers at the beginning,
but would need to move to other influencers when the initially optimal ones
start to lose influence due to fatigue.

\begin{figure}[h]
	\centering
  \includegraphics[width=0.4\textwidth]{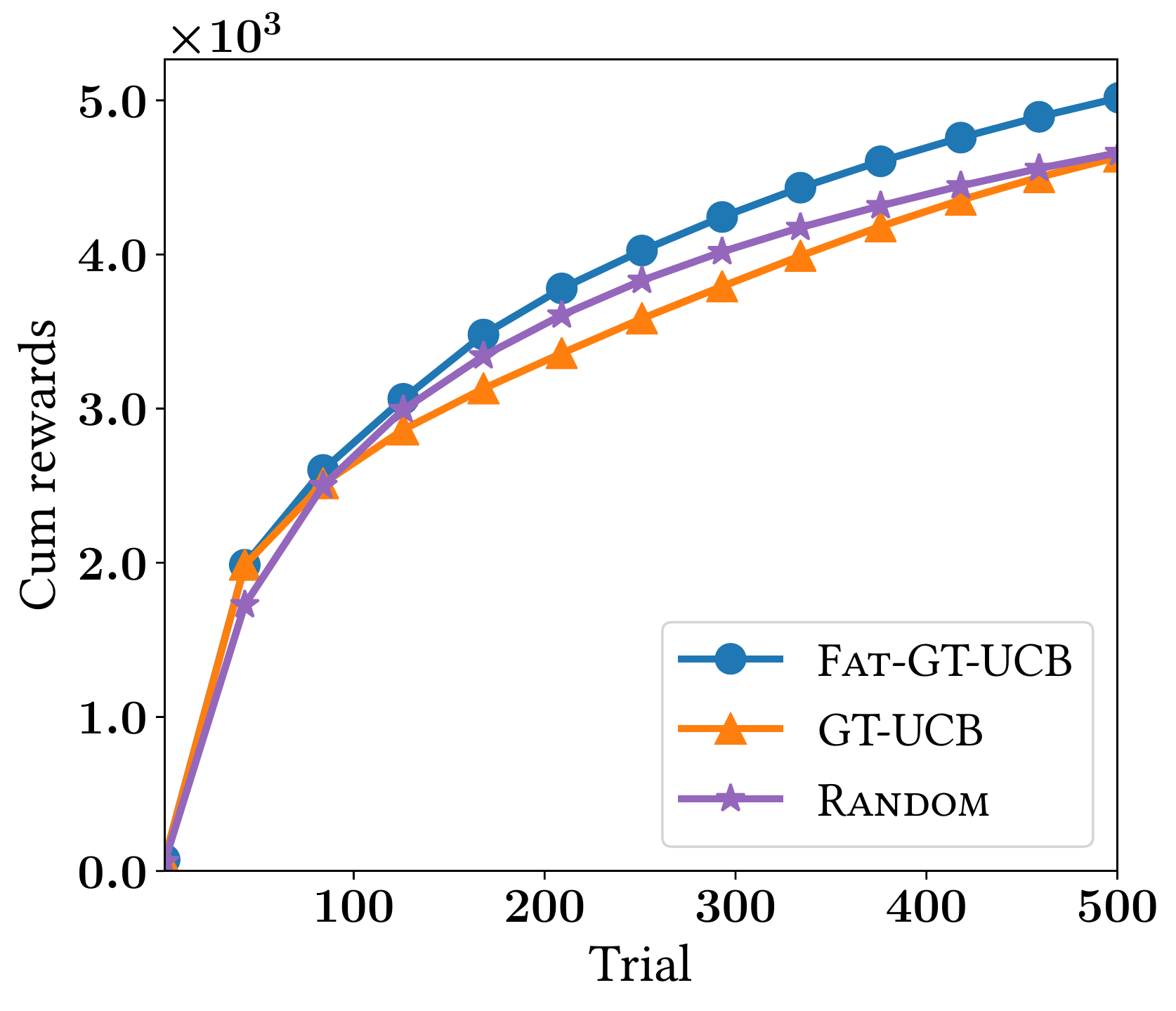}
  \includegraphics[width=0.4\textwidth]{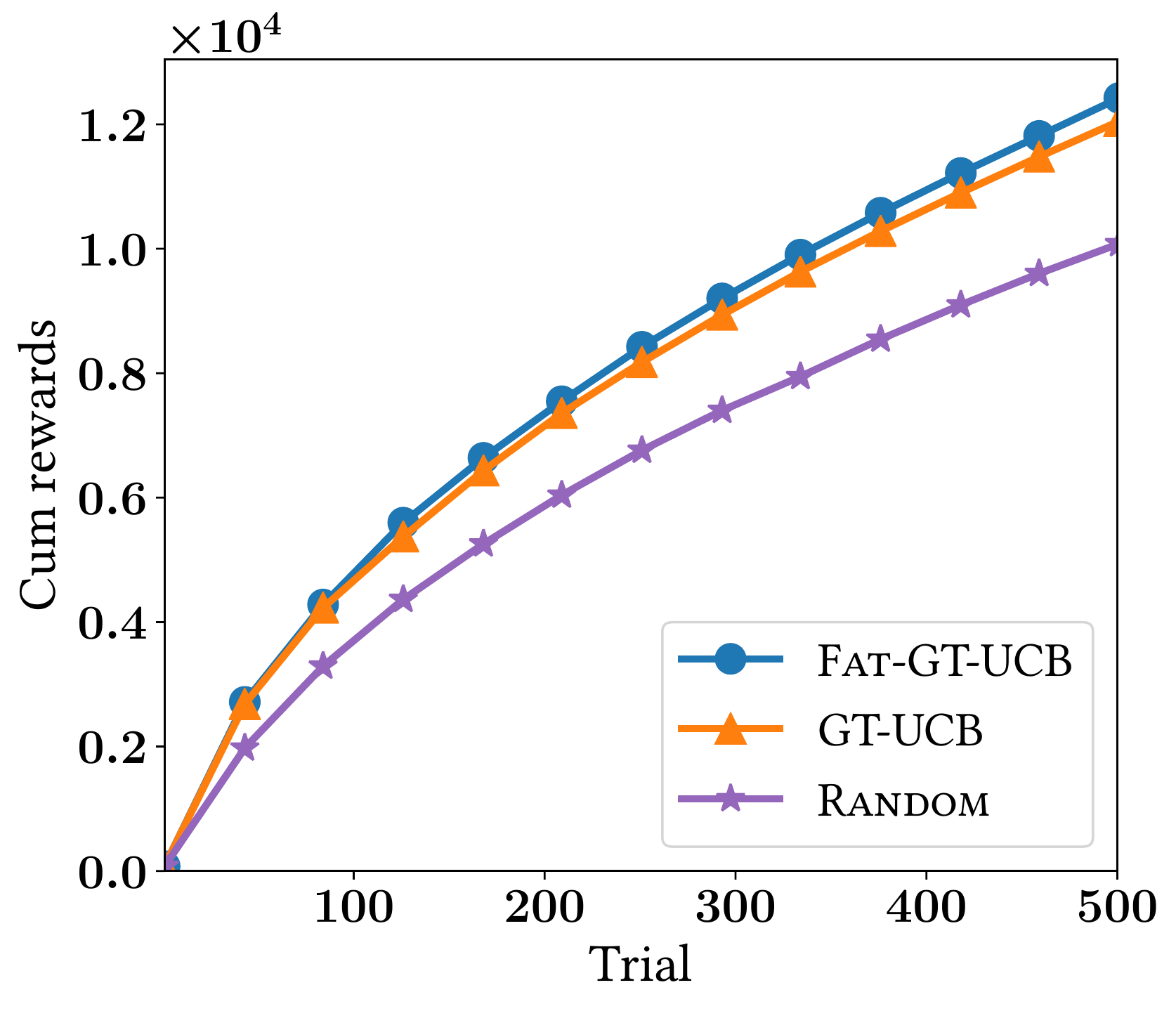}
  \caption{\algonamefat\ vs competitors on Twitter logs with (left)
  $\gamma_1$, (right) $\gamma_2$. \label{fig:fat_gt_ucb}}
\end{figure}

We compare \algonamefat\ to \algoname, which does not use the influence
fatigue function, and to the \textsc{Random} baseline. As in Section~\ref{sec:twitterexp},
when an algorithm selects an influencer, we can choose a random spread from the logs (belonging to the selected influencer), and we can now simulate the fatigue by removing every user in the spread with
probability $\gamma(n)$, where $n$ is the number of times the influencer has
already been played.

We show the results of this comparison in Fig.~\ref{fig:fat_gt_ucb}. We tested
with two different weariness functions, namely $\gamma_1(n) = 1 / n$ and
$\gamma_2(n) = 1 / \sqrt{n}$. We can see that, in both scenarios,
\algonamefat\ performs the best, showing that our UCB-like approach can
effectively handle the notion of influencer fatigue in the OIMP problem.
Unsurprisingly, \algoname\ performs better with weariness function $\gamma_2$
than is does with $\gamma_1$: the former has a lower diminishing impact and
thus, the penalty of not incorporating fatigue is less problematic with
$\gamma_2$.

%
%


%% file: related.tex
We have already discussed in Section~\ref{sec:introduction}  the main related
studies in the area of influence maximization. For further details, we refer the
interested reader to the recent survey in~\cite{arora17}, which discusses the
pros and cons of the best known techniques for influence maximization. In
particular, the authors highlight that the \textit{Weighted Cascade} (WC)
instance of IC, where the weights associated to a node's incoming edges must sum
to one, leads to poor performance for otherwise rather fast IC algorithms. They
conclude that PMC~\cite{ohsaka14} is the state-of-the-art method to efficiently
solve the IC optimization problem, while TIM+~\cite{tang14} and
IMM~\cite{tang15} -- later improved by~\cite{nguyen16} with SSA -- are the best
current algorithms for WC and LT models.



Other methods have been devised to handle the prevalent uncertainty in diffusion
media, e.g., when replacing edge probability scores with ranges thereof,  by
solving an influence maximization problem whose \emph{robust} outcome should
provide some effectiveness guarantees w.r.t. all possible instantiations of the
uncertain model~\cite{he16,chen16}.

Methods for influence maximization that take into account more detailed
information, such as topical categories, have been considered in the
literature~\cite{du13,barbieri13,Wang14}.  Interestingly,~\cite{romero11}
experimentally validates the intuition that different kinds of information
spread differently in social networks, by relying on two complementary
properties, namely  \emph{stickiness} and \emph{persistence}. The former can be
seen as a measure of how viral the piece of information is, passing from one
individual to the next. The latter can be seen as an indicator of the extent to
which repeated exposures to that piece of information impact its adoption, and
it was shown to characterize \emph{complex contagions}, of controversial
information (e.g., from politics).

%% file: conclusion.tex

We propose in this paper a diffusion-independent approach for online and
adaptive influencer marketing, whose role is to maximize the number of activated
nodes in an arbitrary environment, under the OIMP framework. We focus on
scenarios motivated by influencer marketing, in which campaigns consist of
multiple consecutive trials conveying the same piece of information, requiring
as only interfaces with the ``real-world'' the identification of potential seeds
(the influencers) and the spread feedback (i.e., the set of activated nodes) at
each trial. 
Our method's online iterations are very fast, making it possible to scale to
very large graphs, where other approaches become infeasible. The efficiency of
\algoname\ comes from the fact that it only relies on an estimate of a single
quantity for each influencer  --~its remaining potential. This novel approach is
shown to be very competitive on influence maximization benchmark tasks and on
influence spreads in Twitter. Finally, we extend our method to scenarios where,
during a marketing campaign, the influencers may have a \emph{diminishing}
tendency to activate their user base.

%% file: arxivappendix.tex

\section{Useful Lemmas}

\begin{lemma}[Bennett's inequality (Theorem 2.9 and 2.10 \cite{boucheron13})]
  \label{lem:bennett}
  Let $X_1,\ldots,X_n$ be independent random variables with finite variance
  such that $X_i \leq b$ for some $b > 0$ for all $i \leq n$. Let $S :=
  \sum_{i=1}^n \left(X_i - \mathbb{E}[X_i]\right)$ and $v := \sum_{i=1}^n
  \mathbb{E}[X_i^2]$. Writing $\phi(u) = e^u - u - 1$, then for all $t > 0$,
  \begin{align*}
    \log \mathbb{E}\left[e^{tS}\right] \leq \frac{v}{b^2} \phi(bt)
    \leq \frac{vt^2}{2(1 - bt/3)}.
  \end{align*}
  This implies that, $\mathbb{P}\left(S > \sqrt{2v\log \nicefrac{1}{\delta}} +
  \frac{b}{3} \log \nicefrac{1}{\delta} \right) \leq \delta$.
\end{lemma}

\begin{lemma}[Lemma 7 -- \cite{berend13}]\label{lem:berend}
  Let $n \geq 1$, $\lambda \geq 0$, $p \in [0,1]$ and $q = (1 - p)^n$. Then,
  \begin{align}
    qe^{\lambda p(1-q)} + (1-q)e^{-\lambda pq} \leq \exp(p\lambda^2/(4n))\label{eq:berend1}\\
    qe^{\lambda p(q-1)} + (1-q)e^{\lambda pq} \leq \exp(p\lambda^2/(4n))\label{eq:berend2}
  \end{align}
\end{lemma}

\section{Analysis of the Waiting Time of \algoname\ Algorithm}
\label{app:wtanalysis}

\begin{lemma}\label{lem:evolutionestmm}
  For any $s \geq 3$, $\mathbb{P}\left(\hat{R}_{s} \leq \hat{R}_{s-1} -
  \frac{\lambda}{e(s-2)} - \sqrt{\frac{2\lambda}{s-1}\log(1/\delta)} -
  \frac{1}{3(s-1)}\log(1/\delta)\right) \leq \delta.$
\end{lemma}

\begin{proof}
  Denote by $X_s(x) := \frac{U_{s-1}(x)}{s-1} - \frac{U_{s}(x)}{s} \leq
  \frac{1}{s-1}$. We can
  rewrite $\hat{R}_{s-1} - \hat{R}_{s} = \sum_{x \in A} X_s(x)$
  and can easily verify that
  \begin{align}
    v(x) := \mathbb{E}\left[X_s(x)^2\right] = p(x)(1 - p(x))^{s-2}
    \left(\frac{1}{s-1} - \frac{1 - p(x)}{s}\right) \leq \frac{p(x)}{s-1}. \label{eq:mmvx}
  \end{align}
  Let $t > 0$. By applying Lemma\ref{lem:bennett}, one obtains
  \begin{align*}
    \mathbb{P}\left(\hat{R}_{s-1} - \hat{R}_{s} \geq \mathbb{E}\left[\hat{R}_{s-1}
      - \hat{R}_{s}\right] + \sqrt{\frac{2\lambda}{s - 1}\log (1/\delta)} +
      \frac{1}{3(s - 1)}\log (1/\delta)\right) \leq \delta.
  \end{align*}
  We conclude remarking that $\mathbb{E}[X_s(x)] = p(x)^2(1-p(x))^{s-2} \leq
  \frac{p(x)}{e(s-2)}$, that is, $\mathbb{E}[\hat{R}_{s-1} - \hat{R}_{s}] \leq
  \frac{\lambda}{e(s-2)}$.
\end{proof}

\begin{theorem}[Waiting time]
  Denote $\lambda^{\text{min}} := \min_{k \in [K]} \lambda_k$ and
  $\lambda^{\text{max}} := \max_{k \in [K]} \lambda_k$. Assume that
  $\lambda^{\text{min}} \geq 13$. Then, for any $\alpha \in
  \left[\frac{13}{\lambda\text{min}}, 1\right]$, if we define $\tau^* :=
  T^*\left(\alpha - \frac{13}{\lambda^{\text{min}}}\right)$, with probability
  at least $1 - \frac{2K}{\lambda^{\text{max}}}$,
  \begin{align*}
    T_{\text{UCB}}(\alpha) \leq \tau^* + K\lambda^{\text{max}} \log(4\tau^* +
    11K\lambda^{\text{max}}) + 2K.
  \end{align*}
\end{theorem}

\begin{proof}
  Let us define the following confidence bounds:
  \begin{align*}
    b^+_{k,s}(t) &:= (1 + \sqrt{2})\sqrt{\frac{3\lambda_k\log(2t)}{s}} +
    \frac{\log(2t)}{s}, \\
    b^-_{k,s}(t) &:= (1 + \sqrt{2})\sqrt{\frac{3\lambda_k\log(2t)}{s}} +
    \frac{\log(2t)}{s} + \frac{\lambda_k}{s}\text{, and} \\
    c^-_{k,t}(t) &:= \frac{\lambda}{e(s-2)} + \sqrt{\frac{6\lambda_k\log(t)}{s-1}}
    + \frac{\log(t)}{s - 1}.
  \end{align*}
  Let $S > 0$. Using these definitions, we introduce the following events:
  \begin{align*}
    \mathcal{F} &:= \left\{\forall k \in [K], \forall t > S, \forall s \leq t,
    \hat{R}_{k,s} - b^-_{k,s}(t) \leq R_{k,s} \leq \hat{R}_{k,s} + b^+_{k,s}(t)
    \right\}, \\
    \mathcal{G} &:= \left\{\forall k \in [K], \forall s \geq S, \hat{R}_{k,s}
    \geq \hat{R}_{k,s-1} - c^-_{k,s}(t) \right\}, \\
    \mathcal{E} &:= \mathcal{F} \cap \mathcal{G}.
  \end{align*}

  Using Theorem \ref{th:confidence_bounds}, Lemma \ref{lem:evolutionestmm}
  and a union bound, one obtains $\mathbb{P}(\mathcal{E}) \geq 1 - \frac{2K}{S}$
  (by setting $\delta \equiv \frac{1}{t^3}$). Indeed,

  \begin{align*}
    \mathbb{P}\left(\bar{\mathcal{E}}\right) \leq \mathbb{P}(\bar{\mathcal{F}})
        + \mathbb{P}(\bar{\mathcal{G}})
      \leq 2 \sum_{k=1}^K \sum_{t>S} \sum_{s \leq t} \frac{1}{t^3}
      = 2K \sum_{t > S} \frac{1}{t^2} \leq \frac{2K}{S}.
  \end{align*}

  In the following, we work on the event $\mathcal{E}$. Recall that we want to
  control $T_{UCB}(\alpha)$, the time at which every influencer attains a
  remaining potential smaller than $\alpha$ following \algoname\ strategy. We aim
  at comparing $T_{UCB}(\alpha)$ to $T^*(\alpha)$, the same quantity following
  the omniscient strategy. With that in mind, one can write:

  \begin{align*}
    &T_{UCB}(\alpha) = \min \left\{t : \forall k \in [K], R_{k,N_k(t)} \leq
    \alpha \lambda_k \right\}, \\
    &T^*(\alpha) = \sum_{k = 1}^K T^*_k(\alpha) \text{, where } T^*_k(\alpha) =
    \min \left\{s : R_{k,s} \leq \alpha \lambda_k \right\}.
  \end{align*}

  Following ideas from \cite{bubeck13}, we can control $T_{UCB}(\alpha)$ by
  comparing it to $U(\alpha)$ defined below, and which replaces the remaining
  potential by an upper bound on the \textit{estimator} of the remaining
  potential (the Good-Turing
  estimator). Indeed, recall that we can control this on event $\mathcal{F}$.
  \[
    U(\alpha) = \min\left\{t \geq 1 : \forall k \in [K], \hat{R}_{k, N_k(t)}
    + b^+_{k,N_k(t)}(t) \leq \alpha \lambda_k \right\}.
  \]
  Let $S'\geq S$. On event $\mathcal{E}$, one has that $T_{UCB}(\alpha) \leq
  \max(S',U(\alpha))$. If $U(\alpha) \geq S'$, one has
  \begin{align*}
    R_{k,N_k(U(\alpha))} &\geq \hat{R}_{k,N_k(U(\alpha))} -
    b^-_{k,N_k(U(\alpha))}(U(\alpha)) \tag*{(we are on event
    $\mathcal{F}$ and $U(\alpha) > S' \geq S$)}  \\
    &\geq \hat{R}_{k,N_k(U(\alpha)) - 1} - b^-_{k,N_k(U(\alpha))}(U(\alpha))
    - c^-_{k,N_k(U(\alpha))}(U(\alpha)) \tag*{(where are on event
    $\mathcal{G}$)} \\
    &\geq \left(\alpha \lambda_k - b^+_{k,N_k(U(\alpha))-1}(U(\alpha)) \right)
    - b^-_{k,N_k(U(\alpha))}(U(\alpha)) - c^-_{k,N_k(U(\alpha))}(U(\alpha))
  \end{align*}
  The third inequality's justification is more evolved. Let $t$ be the time such
  that $N_k(t) = N_k(U(\alpha)) - 1$ and $N_k(t+1) = N_k(U(\alpha))$. This
  implies that $k$ is the chosen expert at time $t$, that is, the one maximizing
  the \algoname\ index. Moreover, since $t < U(\alpha)$, one knows that this index
  is greater than $\alpha \lambda_k$.

  If $N_k(U(\alpha)) \geq S' + 2$, some basic calculations lead to
  \begin{align*}
    R_{k,N_k(U(\alpha))} \geq \alpha \lambda_k - 11 \sqrt{\frac{\lambda_k
    \log(2U(\alpha))}{S'}} - \frac{3\log(2U(\alpha))}{S'} - \frac{3\lambda_k}{2S'}
  \end{align*}
  We denote by $\lambda^{max} := \max_k \lambda_k$. If we take $S' = \lambda^{max}
  \log(2U(\alpha))$, we can rewrite the previous inequality as
  \begin{align*}
    R_{k,N_k(U(\alpha))} \geq \alpha \lambda_k - 11 - \frac{3}{\lambda^{max}}
      - \frac{3}{2}
  \end{align*}
  Thus, by definition of $T^*_k(\alpha)$, and if $\lambda^{max} > 6$, one gets
  \begin{align*}
    N_{k, U(\alpha)} \leq T_k^*\left(\alpha - \frac{13}{\lambda_k} \right) + S' + 2.
  \end{align*}
  Finally, if we denote by $\lambda^{min} = \min_k \lambda_k$, we obtain that
  \begin{align*}
    U(\alpha) &\leq K(S' + 2) + T^*\left(\alpha - \frac{13}{\lambda^{min}} \right).
  \end{align*}

  We now apply Lemma \ref{lem:bubeck3}. We obtain that
  \begin{align*}
    U(\alpha) \leq 2K + \tau^* + K\lambda^{max} \log \left(8K + 4\tau^*
      + 10K\lambda^{max}\right) \leq \tau^* + K\lambda^{max} \log \left(4\tau^*
      + 11K\lambda^{max}\right) + 2K .
  \end{align*}
  We conclude with $T_{UCB}(\alpha) \leq \max(S', U(\alpha))$.
\end{proof}

\begin{lemma}[Lemma 3 from \cite{bubeck13}] \label{lem:bubeck3}
  Let $a > 0$, $b \geq 0.4$, and $x \geq e$, such that $x \leq a + b \log x$.
  Then one has
  \begin{align*}
    x \leq a + b \log(2a + 4b\log(4b)) .
  \end{align*}
  Moreover, we add that if $b \geq 3$, then $x \leq a + b \log (2a + 5b)$.
\end{lemma}

\section{Confidence intervals in the influencer fatigue setting}
\label{sec:appendixfatigue}

In this section, we consider a single influencer and omit its index $k$. We
recall that we make the assumption that influencers have non-intersecting
support. Thus, after selecting the influencer $n$ times, the remaining potential
can be rewritten
$$
  R_n = \sum_{u \in A} \mathds{1}\{u \text{ never activated }\} p_{n+1}(u),
$$
--~$n+1$ because this is the remaining potential for the $n+1$th spread~-- and
the corresponding Good-Turing estimator is
$$
  \hat{R}_n = \frac{1}{n} \sum_{u \in A} U^{\gamma}_n(u),
$$
where $U^{\gamma}_n(u) = \sum_{i = 1}^n \mathds{1}\{X_1 = \ldots = X_{i-1} =
X_{i+1} = \ldots = X_n = 0, X_i = 1\} \frac{\gamma(n+1)}{\gamma(i)}$.

\paragraph*{Estimator bias.}
Lemma~\ref{lem:rottingbias} shows that the estimator of the remaining potential
for the influencer fatigue setting is hardly biased.
\begin{lemma}\label{lem:rottingbias}
  Denoting $\lambda = \sum_{u \in A} p(u)$, the bias of the remaining potential
  estimator is
  $$ \mathbb{E}[R_n] - \mathbb{E}[\hat{R}_n] \in
    \left[-\gamma(n+1)\frac{\lambda}{n},0\right].
  $$
\end{lemma}

\begin{proof}
  We have that
  $$\mathbb{E}[U^\gamma_n(u)] = \sum_{i = 1}^n p_i(u) \prod_{j \neq i} (1 -
  p_j(u)) \frac{\gamma(n+1)}{\gamma(i)} = p_{n+1}(u) \sum_{i = 1}^n \prod_{j \neq i}
  (1 - p_j(u)).
  $$
  We now can compute the bias of the estimator:
  \begin{align*}
    \mathbb{E}[R_n] - \mathbb{E}[\hat{R}_n]
      &= \frac{1}{n} \sum_{u \in A} p_{n+1}(u) \left[\sum_{i = 1}^n \prod_{j = 1}^n (1 -
        p_j(u)) - \sum_{i = 1}^n \prod_{j \neq i} (1 - p_j(u)) \right] \\
      &= \frac{1}{n} \sum_{u \in A} p_{n+1}(u) \sum_{i = 1}^n \prod_{j \neq i} (1 -
        p_j(u)) [1 - p_i(u) - 1] \\
      &= -\frac{1}{n} \sum_{u \in A} p_{n+1}(u) \sum_{i = 1}^n p_i(u) \prod_{j \neq i}
        (1 - p_j(u)) \\
      &= -\frac{1}{n} \mathbb{E}\left[\sum_{u \in A} p_{n+1}(u)U_n(u) \right]
        \in \left[-\frac{\sum_{u\in A}p_{n+1}(u)}{n}, 0\right]
  \end{align*}
  Note that the random variable $U_n(u)$ correspond to the hapax definition
  given in the original OIMP problem, that is, $U_n(u) = \mathds{1}\{u
  \text{ activated exactly once}\}$.
\end{proof}

Unsurprisingly, we obtain the same bias for the case where $\gamma$ is constant
equal to 1 (no fatigue).

\paragraph*{Confidence Intervals.}
To derive an optimistic algorithm, we need confidence intervals on the remaining
potential. We operate in three steps:

\begin{enumerate}
  \item \textbf{Good-Turing deviations:} Remember that $\hat{R}_n = \sum_{u \in
    A} \frac{U_n^{\gamma}(u)}{n}$. We have next that
    \begin{align*}
      \mathbb{E}[U_n^{\gamma}(u)^2] &= \sum_{i = 1}^n p_i(u) \prod_{j \neq i}
          (1 - p_j(u)) \frac{\gamma(n+1)^2}{\gamma(i)^2} \\
        &= \sum_{i = 1}^n p_{n+1}(u) \prod_{j \neq i} (1 - p_j(u))
          \frac{\gamma(n+1)}{\gamma(i)} \\
        &\leq np(u)\gamma(n+1).
    \end{align*}
    Thus, we have that $v := \sum_{u \in A}
    \mathbb{E}\left[\frac{U^{\gamma}_n(u)^2}{n^2}\right] \leq
    \frac{\sum_{u \in A} p_{n+1}(u)}{n}$.

    Applying Bennett's inequality to the independent random variables
    $\{X^{\gamma}_n(u)\}_{u \in A}$ yields
    \begin{align}
      \mathbb{P}\left(\hat{R}_n - \mathbb{E}\left[\hat{R}_n\right] \geq
      \sqrt{\frac{2\lambda_{n+1} \log(1 / \delta)}{n}} + \frac{1}{3n} \log(1/\delta)
      \right) \leq \delta,
    \end{align}
    where $\lambda_n := \gamma(n)\sum_{u\in A} p(u)$.

    The same inequality can be derived for left deviations.

  \item \textbf{Remaining potential deviations:} Remember that $R_n = \sum_{u \in
    A} Z_n(u) p_{n+1}(u)$ where $Z_n(u) = \mathds{1}\{u \text{ never
    activated}\} = \mathds{1}\{X_1(u) = \cdots = X_n(u) = 0\}$. We denote
    $Y_n(u) = p_{n+1}(u)(Z_n(u) - \mathbb{E}[Z_n(u)])$ and $q_n(u) =
    \mathbb{P}(Z_n(u) = 1) = \prod_{i = 1}^n (1 - p_s(u))$. We have that
    \begin{align*}
      \mathbb{P}\left(R_n - \mathbb{E}[R_n] \geq \epsilon\right) &\leq
          e^{-t\epsilon}\prod_{u \in A}
          \mathbb{E}\left[e^{tY_n(u)}\right]  \\
        &= e^{-t\epsilon}\prod_{u \in A} \left(\mathbb{P}(Z_n(u)
          = 1) e^{tp_{n+1}(u)(1 - q_n(u))} + \mathbb{P}(Z_n(u) = 0)
          e^{-tp_{n+1}(u)q_n(u)}\right) \\
        &= e^{-t\epsilon} \prod_{u \in A} \left(q_n(u) e^{tp_{n+1}(u)(1 -
        q_n(u))} + (1 - q_n(u)) e^{-tp_{n+1}(u)q_n(u)}\right) \\
        &\leq e^{-t\epsilon} \prod_{u \in A} \exp\left(\frac{p_{n+1}(u)
        t^2}{4n}\right) \tag*{(by Eq.~\ref{eq:dev1} in Lemma~\ref{lemma:dev})}
    \end{align*}
    Minimizing on $t$, we obtain (for $t = \frac{2\epsilon n}{\sum_{u \in
    A} p_{n+1}(u)})$,
    $$
      \mathbb{P}\left(R_n - \mathbb{E}[R_n] \geq \epsilon\right) \leq
      \exp\left(\frac{-\epsilon^2n}{\sum_{u \in A}p_{n+1}(u)}\right).
    $$
    We can proceed similarly to obtain the left deviation.
\end{enumerate}

\vspace{2mm}
Putting it all together, we obtain the following confidence intervals in
Theorem~\ref{th:confidence_bounds_fatigue}, which can be used in the design of the
optimistic algorithm \algonamefat.
 
\begin{theorem}\label{th:confidence_bounds_fatigue}
  With probability at least $1 - \delta$,  for $\lambda_n = \gamma(n) \sum_{u \in
  A} p(u)$ and $\beta_n := \left(1 + \sqrt{2}\right) \times
  \sqrt{\frac{\lambda_{n+1} \log(4/\delta)}{n}} + \frac{1}{3n}\log\frac{4}{\delta}$,
  the following holds:
  \[
    - \beta_n - \frac{\lambda}{n} \leq R_n - \hat{R}_n \leq \beta_n.
  \]
\end{theorem}

\begin{lemma}[Adaptation of Lemma 3.5 in~\cite{berend13}\label{lemma:dev}]
  Let $n \geq 1$, $p \in [0, 1]$, $\gamma: \mathbb{N} \to [0,1]$ a
  non-increasing function and $t \geq 0$. We denote $p_n = \gamma(n)p$ and $q_n
  = \prod_{i \leq t} (1 - p_i)$.
  \begin{align}
    &(a)~~q_n e^{tp_n(1 - q_n)} + (1 - q_n) e^{-tp_nq_n} \leq
      \exp\left(\frac{p_nt^2}{4n}\right)  \label{eq:dev1}\\
    &(b)~~q_n e^{tp_n(q_n - 1)} + (1 - q_n) e^{tp_nq_n} \leq
      \exp\left(\frac{p_nt^2}{4n}\right)
  \end{align}
\end{lemma}

\begin{proof}
  Let $q_n' = (1 - p_n)^n$. Clearly, $q_n \leq q_n'$.

  \vspace{2mm}
  \noindent(a) Using Theorem 3.2 in~\cite{berend13} with $p \equiv q_n$ and $t \equiv tp_n$, we have
  that,
  $$
    q_n e^{tp_n(1 - q_n)} + (1 - q_n) e^{-tp_nq_n} \leq
    \exp\left(\frac{1 - 2q_n}{4\log((1-q_n)/q_n)}t^2p_n^2\right).
  $$
  So it suffices to show that
  $$
    (1 - 2q_n)t^2p_n^2/4\log((1-q_n)/q_n) \leq p_nt^2/4n,
  $$
  or equivalently,
  $$
    (1 - 2q_n)p_n\log((1-q_n)/q_n) \leq \log(1 - p_n) / \log(q_n').
  $$
  Rearranging, we obtain,
  $$
    L(q_n, q_n') := \frac{(1 - 2q_n)\log(1/q_n')}{\log((1-q_n)/q_n)} \leq
    \frac{1 / \log(1 - p_n)}{p_n} := R(p_n).
  $$
  As is \cite{berend13}, we show that $L \leq 1 \leq R$. The second inequality
  is true (see \cite{berend13}).

  The left-hand side can be written $L(q_n, q_n') = L_1(q_n) L_2(q_n')$.
  Clearly, $L_2(q_n') \geq 0$. If $L_1(q_n) \leq 0$, the left-hand side is
  negative and thus, $L(q_n, q_n') \leq 1$. Else, $L_1(q_n) \geq 0$, we can
  upper bound the right-hand side as $L(q_n, q_n') \leq \frac{(1 -
  2q_n)\log(1/q_n)}{\log((1-q_n)/q_n)}$ (because $q_n \leq q_n'$), which is
  proven to be less than $1$ in \cite{berend13}. This concludes the proof.

  \vspace{2mm}\noindent(b) It is shown in the proof Lemma 3.5 (b)
  in~\cite{berend13} that
  $$
    L(t) := \frac{1}{t^2p_n^2} \log\left[q_ne^{-tp_n(1-q_n)} + (1 - q_n)e^{tp_n
    q_n}\right] \leq \frac{1}{t^2p_n^2}\frac{t^2p_n}{4\log q_n/\log(1-p_n)} =:
    R(q_n).
  $$
  It suffices to show that $R(q_n) \leq R(q_n')$ to obtain the desired
  inequality. This is true because
  $$
    0 \leq \frac{\log(1 - p_n)}{\log q_n} \leq \frac{\log(1 - p_n)}{\log q_n'}.
  $$
\end{proof}

%% file: arxiv.bbl
\begin{thebibliography}{10}

\bibitem{duncan08}
In Duncan Brown and Nick Hayes, editors, {\em Influencer Marketing}.
  Butterworth-Heinemann, Oxford, 2008.

\bibitem{arora17}
A.~Arora, S.~Galhotra, and S.~Ranu.
\newblock Debunking the myths of influence maximization: An in-depth
  benchmarking study.
\newblock In {\em SIGMOD}. ACM, 2017.

\bibitem{barbieri13}
N.~Barbieri, F.~Bonchi, and G.~Manco.
\newblock Topic-aware social influence propagation models.
\newblock {\em Knowl. Inf. Syst.}, 37(3):555--584, 2013.

\bibitem{berend13}
D.~Berend and A.~Kontorovich.
\newblock On the concentration of the missing mass.
\newblock In {\em Electronic Communications in Probability}, pages 1--7, 2013.

\bibitem{boucheron13}
S.~Boucheron, G.~Lugosi, P.~Massart, and M.~Ledoux.
\newblock {\em Concentration inequalities : a nonasymptotic theory of
  independence}.
\newblock Oxford university press, 2013.

\bibitem{brown13}
D.~Brown and S.~Fiorella.
\newblock {\em Influence Marketing: How to Create, Manage, and Measure Brand
  Influencers in Social Media Marketing}.
\newblock Always learning. Que, 2013.

\bibitem{bubeck12}
S.~Bubeck and N.~Cesa{-}Bianchi.
\newblock Regret analysis of stochastic and nonstochastic multi-armed bandit
  problems.
\newblock {\em Foundations \& Trends in ML}, 2012.

\bibitem{bubeck13}
S.~Bubeck, D.~Ernst, and A.~Garivier.
\newblock Optimal discovery with probabilistic expert advice: finite time
  analysis and macroscopic optimality.
\newblock {\em Journal of Machine Learning Research}, 14(1):601--623, 2013.

\bibitem{chen16}
W.~Chen, T.~Lin, Z.~Tan, M.~Zhao, and X.~Zhou.
\newblock Robust influence maximization.
\newblock In {\em SIGKDD}, pages 795--804, 2016.

\bibitem{chen16-2}
W.~Chen, Y.~Wang, Y.~Yuan, and Q.~Wang.
\newblock Combinatorial multi-armed bandit and its extension to
  probabilistically triggered arms.
\newblock {\em JMLR}, 17(1), 2016.

\bibitem{du13}
N.~Du, L.~Song, H.~Woo, and H.~Zha.
\newblock Uncover topic-sensitive information diffusion networks.
\newblock In {\em AISTATS}, pages 229--237, 2013.

\bibitem{easley10}
D.~Easley and J.~Kleinberg.
\newblock {\em Networks, Crowds, and Markets - Reasoning About a Highly
  Connected World}.
\newblock Cambridge University Press, 2010.

\bibitem{gillin07}
Paul Gillin.
\newblock {\em The New Influencers: A Marketer's Guide to the New Social
  Media}.
\newblock Quill Driver Books, Sanger, CA, 2007.

\bibitem{gomez11}
M.~Gomez{-}Rodriguez, D.~Balduzzi, and B.~Sch{\"{o}}lkopf.
\newblock Uncovering the temporal dynamics of diffusion networks.
\newblock In {\em ICML}, pages 561--568, 2011.

\bibitem{gomez12}
M.~Gomez-Rodriguez, J.~Leskovec, and A.~Krause.
\newblock Inferring networks of diffusion and influence.
\newblock {\em ACM Trans. Knowl. Discov. Data}, 5(4), February 2012.

\bibitem{gomez13}
M.~Gomez{-}Rodriguez, J.~Leskovec, and B.~Sch{\"{o}}lkopf.
\newblock Structure and dynamics of information pathways in online media.
\newblock In {\em WSDM}, pages 23--32, 2013.

\bibitem{good53}
I.~J. Good.
\newblock The population frequencies of species and the estimation of
  population parameters.
\newblock {\em Biometrika}, 40(3-4):237, 1953.

\bibitem{goyal10}
A.~Goyal, F.~Bonchi, and L.~Lakshmanan.
\newblock Learning influence probabilities in social networks.
\newblock In {\em WSDM}, pages 241--250, 2010.

\bibitem{grabowicz16}
P.~Grabowicz, N.~Ganguly, and K.~Gummadi.
\newblock Distinguishing between topical and non-topical information diffusion
  mechanisms in social media.
\newblock In {\em ICWSM}, pages 151--160, 2016.

\bibitem{he16}
X.~He and D.~Kempe.
\newblock Robust influence maximization.
\newblock In {\em SIGKDD}, pages 885--894, 2016.

\bibitem{kempe03}
D.~Kempe, J.~Kleinberg, and \'{E}. Tardos.
\newblock Maximizing the spread of influence through a social network.
\newblock In {\em SIGKDD}, pages 137--146. ACM, 2003.

\bibitem{lagree17}
Paul Lagr\'ee, Olivier Capp\'e, Bogdan Cautis, and Silviu Maniu.
\newblock Effective large-scale online influence maximization.
\newblock In {\em ICDM}, 2017.

\bibitem{lee17}
Kamiu Lee.
\newblock Influencer marketing 2.0: Key trends in 2017.
\newblock
  https://influence.bloglovin.com/influencer-marketing-2-0-key-trends-in-2017-a5fb97424cd,
  2017.

\bibitem{lei15}
S.~Lei, S.~Maniu, L.~Mo, R.~Cheng, and P.~Senellart.
\newblock Online influence maximization.
\newblock In {\em SIGKDD}, 2015.

\bibitem{levine17}
Nir Levine, Koby Crammer, and Shie Mannor.
\newblock Rotting bandits.
\newblock In {\em Advances in Neural Information Processing Systems 30 (NIPS)},
  2017.

\bibitem{louedec16}
Jonathan Lou{\"e}dec, Laurent Rossi, Max Chevalier, Aur{\'e}lien Garivier, and
  Josiane Mothe.
\newblock Algorithme de bandit et obsolescence : un mod{\`e}le pour la
  recommandation.
\newblock 2016.

\bibitem{mcallester03}
D.~McAllester and L.~Ortiz.
\newblock Concentration inequalities for the missing mass and for histogram
  rule error.
\newblock {\em JMLR}, 4:895--911, 2003.

\bibitem{mcallester00}
D.~McAllester and R.~Schapire.
\newblock On the convergence rate of good-turing estimators.
\newblock In {\em COLT}, pages 1--6, 2000.

\bibitem{mei10}
Q.~Mei, J.~Guo, and D.~Radev.
\newblock Divrank: The interplay of prestige and diversity in information
  networks.
\newblock In {\em SIGKDD}, 2010.

\bibitem{nguyen16}
H.~T. Nguyen, M.~T. Thai, and T.~N. Dinh.
\newblock Stop-and-stare: Optimal sampling algorithms for viral marketing in
  billion-scale networks.
\newblock In {\em SIGMOD}, 2016.

\bibitem{ohsaka14}
N.~Ohsaka, T.~Akiba, Y.~Yoshida, and K.~Kawarabayashi.
\newblock Fast and accurate influence maximization on large networks with
  pruned monte-carlo simulations.
\newblock In {\em AAAI}, 2014.

\bibitem{romero11}
D.~Romero, B.~Meeder, and J.~Kleinberg.
\newblock Differences in the mechanics of information diffusion across topics:
  Idioms, political hashtags, and complex contagion on twitter.
\newblock In {\em WWW}, pages 695--704, 2011.

\bibitem{saito08}
K.~Saito, R.~Nakano, and M.~Kimura.
\newblock Prediction of information diffusion probabilities for independent
  cascade model.
\newblock In {\em KES}, pages 67--75, 2008.

\bibitem{sletten17}
Cole Sletten.
\newblock How to prepare brands for influencer fatigue.
\newblock
  https://www.forbes.com/sites/onmarketing/2017/01/24/how-to-prepare-brands-for-influencer-fatigue,
  2017.

\bibitem{tang15}
Y.~Tang, Y.~Shi, and X.~Xiao.
\newblock Influence maximization in near-linear time: A martingale approach.
\newblock In {\em SIGMOD}, pages 1539--1554, 2015.

\bibitem{tang14}
Y.~Tang, X.~Xiao, and Y.~Shi.
\newblock Influence maximization: Near-optimal time complexity meets practical
  efficiency.
\newblock In {\em SIGMOD}, pages 75--86, 2014.

\bibitem{vaswani17}
S.~Vaswani, B.~Kveton, Z.~Wen, M.~Ghavamzadeh, L.~Lakshmanan, and M.~Schmidt.
\newblock Diffusion independent semi-bandit influence maximization.
\newblock In {\em ICML}, 2017.

\bibitem{vaswani15}
S.~Vaswani, V.S. Lakshmanan, and M.~Schmidt.
\newblock Influence maximization with bandits.
\newblock In {\em Workshop NIPS}, NIPS '15, 2015.

\bibitem{Wang14}
S.~Wang, X.~Hu, P.~Yu, and Z.~Li.
\newblock Mmrate: inferring multi-aspect diffusion networks with multi-pattern
  cascades.
\newblock In {\em SIGKDD}, pages 1246--1255, 2014.

\bibitem{watts2003}
D.~Watts.
\newblock {\em Six Degrees: The Science of a Connected Age}.
\newblock W. W. Norton, NY, 2003.

\bibitem{watts07}
D.~Watts and P.~Dodds.
\newblock Influentials, networks, and public opinion formation.
\newblock {\em Journal of Consumer Research}, 34(4):441--458, 2007.

\bibitem{wen16}
Z.~Wen, B.~Kveton, and M.~Valko.
\newblock Influence maximization with semi-bandit feedback.
\newblock Technical report, 2016.

\end{thebibliography}
